\documentclass[a4paper, 11pt]{article}

\usepackage{authblk}
\usepackage{subcaption}
\usepackage{graphicx}
\usepackage{rotating}
\usepackage{amsmath}
\usepackage{algorithm2e}
\usepackage{natbib}
\usepackage{hyperref}

\RestyleAlgo{ruled}
\usepackage{soul}
\usepackage{amssymb}
\usepackage{enumitem}
\usepackage{booktabs}
\usepackage{lscape}
\usepackage{amsthm}
\theoremstyle{definition}
\newtheorem{definition}{Definition}[section]
\newtheorem{conjecture}{Conjecture}
\newtheorem{lemma}{Lemma}
\newtheorem{theorem}{Theorem}

\let\cite\citep

\newcommand{\mset}[1]{\left\{\!\!\left\{#1\right\}\!\!\right\}}

\usepackage{mod}
\usetikzlibrary{arrows, calc, decorations, decorations.pathmorphing, matrix, shapes}
\usetikzlibrary{arrows.meta}
\usetikzlibrary{through}
\tikzset{
	abstractNetwork/.style={font=\scriptsize, scale=0.7},
	vertexBase/.style={circle, minimum size=1em},
	vertex/.style={vertexBase, circle, draw=black},
	hyperEdge/.style={rectangle, draw},
	arrow/.style={semithick, -{Stealth[round]}},
	inputMol/.style={draw=ForestGreen, thick},
	outputMol/.style={draw=NavyBlue, thick},
	borrowMol/.style={draw=Purple, thick},
	inoutMol/.style={draw=Turquoise, thick}
}
\tikzset{modStyleDGHyperEdge/.style={hyperEdge}}

\newcommand\dpoDerivationRaw[6]{%
    \begin{tikzpicture}[gCat, every node/.style={inner sep=1pt, font=\scriptsize}, vCat/.append style={inner sep=.3333em}, node distance=1em and 1em]
       \node[vCat, label=below:$D$]                (D) {#5};
       \node[vCat, label=below:$G$, left=of D]     (G) {#4};
       \node[vCat, label=below:$H$, right=of D]    (H) {#6};
       \draw[eMorphism] (D) to (G);
       \draw[eMorphism] (D) to (H);
       \node[vCat, label=above:$K$, above=of D]    (K) {#2};
       \node[vCat, label=above:$L$, at=(K-|G)]     (L) {#1};
       \node[vCat, label=above:$R$, at=(K-|H)]     (R) {#3};
       \draw[eMorphism] (K) to (L);
       \draw[eMorphism] (K) to (R);

       \draw[eMorphism] (L) to (G);
       \draw[eMorphism] (K) to (D);
       \draw[eMorphism] (R) to (H);
   \end{tikzpicture}%
}

\newcommand\dpoDerivationTwo[3][]{%
	\dpoDerivationRaw%
		{\includegraphics[#1]{#2_derL.pdf}}	{\includegraphics[#1]{#2_derK.pdf}}	{\includegraphics[#1]{#2_derR.pdf}}%
		{\includegraphics[#1]{#3_derG.pdf}}	{\includegraphics[#1]{#3_derD.pdf}}	{\includegraphics[#1]{#3_derH.pdf}}%
}

\usepackage{url}
\makeatletter
\AtBeginDocument{%
  \def\doi#1{\url{https://doi.org/#1}}}
\makeatother

\providecommand{\keywords}[1]
{
  \textbf{\textit{Keywords---}} #1
}
\date{}
\title{Pathway Realisability in Chemical Networks}
\author[1]{Jakob L.\ Andersen}
\author[1$\ast$]{Sissel Banke}
\author[1]{Rolf Fagerberg}
\author[3,9]{Christoph Flamm}
\author[2,1]{Daniel Merkle}
\author[3-8]{Peter F.\  Stadler}
\affil[1]{Department of Mathematics and Computer Science, University of Southern Denmark, Odense M DK-5230, Denmark}
\affil[2]{Algorithmic Cheminformatics Group, Faculty of Technology, Bielefeld University, Bielefeld D-33615, Germany}
\affil[3]{Institute for Theoretical Chemistry, University of Vienna, Wien A-1090, Austria}
\affil[4]{Department of Computer Science, and Interdisciplinary Center for Bioinformatics, University of Leipzig, Leipzig D-04107, Germany}
\affil[5]{Max Planck Institute for Mathematics in the Sciences, Leipzig D-04103, Germany}
\affil[6]{Fraunhofer Institute for Cell Therapy and Immunology, Leipzig D-04103, Germany}
\affil[7]{Center for non-coding RNA in Technology and Health, University of Copenhagen, Frederiksberg C DK-1870, Denmark}
\affil[8]{Santa Fe Institute, 1399 Hyde Park Rd, Santa Fe NM 87501, USA}
\affil[9]{Research Network Chemistry Meets Microbiology, University of Vienna, Wien A-1090, Austria}
\affil[$\ast$]{Corresponding author. E-mail:banke@imada.sdu.dk}

\begin{document}
\maketitle
\begin{abstract}
  The exploration of pathways and alternative pathways that have a
  specific function is of interest in numerous chemical contexts.  A
  framework for specifying and searching for pathways has previously
  been developed, but a focus on which of the many pathway solutions
  are realisable, or can be made realisable, is missing. Realisable
  here means that there actually exists some sequencing of the
  reactions of the pathway that will execute the pathway. We present a
  method for analysing the realisability of pathways based on the
  reachability question in Petri nets. For realisable pathways, our
  method also provides a certificate encoding an order of the
  reactions which realises the pathway. We present two extended
  notions of realisability of pathways, one of which is related to the
  concept of network catalysts.  We exemplify our findings on the
  pentose phosphate pathway.  Furthermore, we discuss the relevance of our
  concepts for elucidating the choices often implicitly made when
  depicting pathways. Lastly, we lay the foundation for the 
  mathematical theory of realisability.
\end{abstract}

\section{Introduction}

Large Chemical Reaction Networks (CRNs) are fundamental to numerous scientific, industrial, and societal challenges. Applications include the analysis of metabolic networks and their regulation in health and biotechnology; optimization of chemical synthesis processes; modelling of molecular ion fragmentation in mass spectrometry; investigation of hypotheses concerning the origins of life; and environmental monitoring of pollutants. Subnetworks with specific properties, often referred to as \emph{pathways}---such as synthetic routes to target molecules or metabolic subsystems---are of particular interest. Thus, the ability to define and identify pathways within a CRN is a central objective in chemical modelling, exploration, and design.

CRNs can be modelled as directed hypergraphs~\cite{Zeigarnik:00,Muller:22,Andersen:20,flow}, where vertices represent molecules and directed hyperedges represent reactions. By considering pathways in CRNs as sets of reactions with integer multiplicities, \cite{flow} formally defined pathways as integer hyperflows in hypergraphs. The integer hyperflow model for pathways is analogous to flux balance analysis (FBA), another method for pathway discovery. Both approaches enforce mass conservation and typically employ linear constraints to identify pathways. However, they differ in several respects; see \cite{flow} for a detailed discussion. Notably, FBA yields flux distributions, whereas integer hyperflows provide pathways as sets of reactions with integer stoichiometric coefficients, facilitating a mechanistic understanding of the pathway. Additionally, \cite{flow} introduced the concept of a \emph{chemical transformation motif} in a CRN, offering a flexible framework for querying reaction networks for pathways. A chemical transformation motif specifies a pathway by prescribing the input and output compounds, allowing intermediate products that must be consumed entirely. Computationally, finding and enumerating pathways that fulfil a chemical transformation motif can be addressed via Integer Linear Programming (ILP)~\cite{flow}. Although ILP is NP-hard in general and even in the restricted context of finding integer hyperflows in CRNs~\cite{NPflow}, current ILP solvers perform well for many practically relevant networks and pathways~\cite{flow}.

While integer hyperflows specify reactions and their multiplicities, they do not determine the sequence in which these reactions occur to achieve the overall chemical transformation. Indeed, a valid sequencing may not exist. Figure~\ref{fig:missing_catalyst} illustrates such a scenario: no ordering of reactions $e_1$ and $e_2$ in the integer hyperflow renders it executable---essentially, molecules $C$ or $D$ must be present prior to their production. We introduce the term \emph{realisable} for integer hyperflows where the corresponding chemical transformation can be executed by some sequence of constituent reactions. To address this, we develop a framework that converts integer hyperflows into corresponding Petri nets, enabling the application of Petri net methodologies to express and determine the realisability of integer hyperflows. Petri nets have been extensively employed to model various aspects of metabolic networks~\cite{Baldan:10}.

For realisable integer hyperflows, we introduce the concept of a \emph{realisability certificate}, which specifies an execution order for the reactions along the pathway. Determining an explicit sequence not only enhances mechanistic understanding but is also essential for studies where individual atom identities are important, such as computing atom traces~\cite{trace}. We also explore methods to extend non-realisable integer hyperflows into realisable ones. One approach involves scaling the integer hyperflow, while another entails borrowing additional molecules that are subsequently returned. This latter method is closely related to the concept of a ``network catalyst'' (see e.g.\ \cite{Braakman:13,Morowitz:08}). An algorithmic approach to deciding realisability through borrowing thus serves as a crucial foundation for future computational treatments of higher-level chemical motifs like autocatalysis and hypercycles~\cite{hypercycle,hypercycleA,Szathmary:1988,Szathmary:2013}. Finally, we apply our methodology to the non-oxidative phase of the pentose phosphate pathway (PPP) to demonstrate its utility and to explore potential catalysts within the network. The PPP is a well-known example that underscores the importance of simplicity in solution finding~\cite{Noor2010,Melendez-Hevia:1985}.

The primary focus of our paper is the formal definition and exploration of the realisability of pathways. It is noteworthy that conventional representations of pathways in the life sciences literature often reside between the two extremes of integer hyperflows and realisability certificates. We believe that our formalisation of these concepts can raise awareness of the implicit choices made when depicting pathways. This perspective is further elaborated in Section~\ref{sec:pathwayrep}.

The remainder of this paper is organised as follows. Section~\ref{sec:preliminaries} presents the notation and definitions for directed hypergraphs, integer hyperflows, and Petri nets, with terminology following~\cite{flow}. Section~\ref{sec:realisability} defines the realisability problem, outlines our method for converting integer hyperflows into Petri nets, and introduces realisability certificates. In Section~\ref{sec:extended_real}, we investigate methods for rendering non-realisable integer hyperflows realisable, either by scaling the hyperflow or by borrowing molecules. Section~\ref{sec:pathwayrep} discusses the implications of integer hyperflows and realisability certificates in pathway depiction. Section~\ref{sec:math_realisability} examines the mathematical properties of pathway realisability.

\section{Preliminaries}\label{sec:preliminaries}
\subsection{Chemical Reaction Networks and Pathways}
\label{subsec:flowpaperrecap}

In this paper we follow \cite{flow} and model CRNs as directed hypergraphs. A directed hypergraph $\mathcal{H} = (V, E)$ has a set $V$ of vertices representing the molecules.
Reactions are represented as directed hyperedges $E$, where each edge
$e = (e^+, e^-)$ is an ordered pair of multisets of vertices, i.e.,
$e^+, e^-\subseteq V$.\footnote{When comparing a multiset~$M$ and a
  set~$S$, we view~$M$ as a set. I.e., $M \subseteq S$ holds if every
  element in $M$ is an element of~$S$.} We call $e^+$ the \emph{tail}
of the edge $e$, and $e^-$ the \emph{head}.  In the interest of
conciseness we will refer to directed hypergraphs simply as
hypergraphs, directed hyperedges simply as edges, and CRNs as
networks.  For a multiset $Q$ and an element $q$ we use $m_q(Q)$ to
denote its multiplicity, i.e., the number of occurrences of $q$ in
$Q$.  When denoting multisets we use the notation $\mset{\dots}$,
e.g., $Q = \mset{a, a, b}$ is a multiset with $m_a(Q) = 2$ and
$m_b(Q) = 1$.  For a vertex $v\in V$ and a set of edges $A$ we use
$\delta^+_A(v)$ and $\delta^-_A(v)$ to denote respectively the set of
out-edges and in-edges of $v$ contained in $A$, i.e., the edges in $A$
that have $v$ in their tail and $v$ in their head, respectively.
We note that hypergraph modelling is equivalent to the more common modelling via a bipartite species-reactions graph \cite{fagerberg:13}. Fig.~\ref{fig:bipartite_hypergraph} shows a directed hypergraph in (a) and its equivalent bipartite graph in (b). The hypergraph modelling can be said to provide a sligthly stronger distinction
between molecules and reactions and it forms the basis of
the modelling of hyperflows in \cite{flow}, on which we build.

\begin{figure}[tbp]
  \centering
  \begin{subfigure}[b]{0.49\textwidth}
  \centering
  \begin{tikzpicture}[abstractNetwork]
    \node[vertex] (B)  at (0,-2)   {$B$} ;
    \node[vertex] (C)  at (4,-2)   {$C$} ;
    \node[vertex] (A)  at (0,-4)  {$A$} ;
    \node[vertex] (D)  at (-4,-2)  {$D$} ;
    \node[hyperEdge, label=below:{$e_1$}] (e1) at (2,-2) {};
    \node[hyperEdge, label=below:{$e_2$}] (e2) at (4,-4) {};
    \node[hyperEdge, label=above:{$e_3$}] (e3) at (-2,-2) {};
    \node[hyperEdge, label=right:{$\strut e_4$}] (e4) at (0,-3) {};
    
    \draw[arrow] (B) -- (e1);
    \draw[arrow] (e1) -- (C);
    \draw[arrow] (C) -- (e2);
    \draw[arrow] (e2) -- (A);
    \draw[arrow] (e2) -- (B);
    \draw[arrow] (D) -- (e3);
    \draw[arrow] (e3) to [bend right=15] (B);
    \draw[arrow] (e3) to [bend left=15] (B);
    \draw[arrow] (D) to (e4);
    \draw[arrow] (B) to (e4);
    \draw[arrow] (e4) to (A);
  \end{tikzpicture}%
  \caption{A directed hypergraph.}
  \end{subfigure}
  \hfill
  \begin{subfigure}[b]{0.49\textwidth}
  \centering
  \begin{tikzpicture}[abstractNetwork, node distance = 0.5cm and 1cm]
  \node[vertex] (D) {$D$} ;
    \node[vertex, below = of D] (B) {$B$} ;
    \node[vertex, below = of B] (A) {$A$} ;
    \node[vertex, below = of A] (C) {$C$} ;
    \node[vertex, right = of A] (e1) {$e_1$};
    \node[vertex, right = of C] (e2) {$e_2$};
    \node[vertex, right = of D] (e3) {$e_3$};
    \node[vertex, right = of B] (e4) {$e_4$};
    
    \draw[arrow] (B) -- (e1);
    \draw[arrow] (e1) -- (C);
    \draw[arrow] (C) -- (e2);
    \draw[arrow] (e2) -- (A);
    \draw[arrow] (e2) -- (B);
    \draw[arrow] (D) -- (e3);
    \draw[arrow] (e3) to [bend right=15] (B);
    \draw[arrow] (e3) to [bend left=15] (B);
    \draw[arrow] (D) to (e4);
    \draw[arrow] (B) to (e4);
    \draw[arrow] (e4) to (A);
  \end{tikzpicture}%
  \caption{A bipartite graph.}
  \end{subfigure}%
  \caption{A directed hypergraph in (a) and the corresponding bipartite graph in (b).}
  \label{fig:bipartite_hypergraph}
\end{figure}

To model pathways \cite{flow} 
defines the \emph{extended hypergraph}.
Given a hypergraph
$\mathcal{H} = (V, E)$  the extended hypergraph is
$\overline{\mathcal{H}}=(V,\overline{E})$ with $\overline{E} = E \cup E^-
\cup E^+$, where
\begin{equation}
  E^- = \{e^-_v = (\emptyset, \mset{v}) \mid v\in V\} \qquad
  E^+ = \{e^+_v = (\mset{v}, \emptyset) \mid v\in V\} 
  \label{eq:extendedGraph}
\end{equation}
The hypergraph $\overline{\mathcal{H}}$ has additional ``half-edges''
$e^-_v$ and $e^+_v$, for each $v\in V$.  These explicitly represent
potential input and output channels to and from $\mathcal{H}$, i.e., what
is called exchange reactions in metabolic networks.  An example of an
extended hypergraph is shown in Fig. \ref{fig:hypergraph}.%

\begin{figure}[tbp]
\centering
  \begin{tikzpicture}[abstractNetwork]
    \node[vertex] (B)  at (0,-2)   {$B$} ;
    \node[vertex] (C)  at (4,-2)   {$C$} ;
    \node[vertex] (A)  at (0,-4)  {$A$} ;
    \node[vertex] (D)  at (-4,-2)  {$D$} ;
    
    \node[vertexBase, overlay] (B1) at (0,0) {\phantom{$B$}};
    \node[vertexBase, overlay] (A1) at (-2,-4) {\phantom{$A$}};
    \node[vertexBase, overlay] (C1) at (6,-2) {\phantom{$C$}};
    \node[vertexBase, overlay] (D1) at (-6,-2) {\phantom{$D$}};
    \node[hyperEdge, label=below:{$e_1$}] (e1) at (2,-2) {};
    \node[hyperEdge, label=below:{$e_2$}] (e2) at (4,-4) {};
    \node[hyperEdge, label=above:{$e_3$}] (e3) at (-2,-2) {};
    \node[hyperEdge, label=right:{$\strut e_4$}] (e4) at (0,-3) {};
    
    \draw[arrow] (B) -- (e1);
    \draw[arrow] (e1) -- (C);
    \draw[arrow] (C) -- (e2);
    \draw[arrow] (e2) -- (A);
    \draw[arrow] (e2) -- (B);
    \draw[arrow] (D) -- (e3);
    \draw[arrow] (e3) to [bend right=15] (B);
    \draw[arrow] (e3) to [bend left=15] (B);
    \draw[arrow] (D) to (e4);
    \draw[arrow] (B) to (e4);
    \draw[arrow] (e4) to (A);
    
    \draw[arrow] (B.110) to  (B1.250);
    \draw[arrow] (B1.290) to (B.70);
		
    \draw[arrow] (A.200) to (A1.340);
    \draw[arrow] (A1.20) to (A.160);
    
    \draw[arrow] (D.200) to (D1.340);
    \draw[arrow] (D1.20) to (D.160);
    
    \draw[arrow] (C1.200) to (C.340);
    \draw[arrow] (C.20) to (C1.160);
  \end{tikzpicture}%
  \caption{Example of an extended hypergraph.  It has vertices
    $\{A,B,C,D\}$, edges $\{e_1,e_2,e_3,e_4\}$, and a half-edge to and
    from each vertex. An edge~$e$ is represented by a box with arrows
    to (from) each element in $e^-$ ($e^+$).}
  \label{fig:hypergraph}%
\end{figure}

An \emph{integer hyperflow} is 
an integer-valued function~$f$ on the extended network,
$f \colon \overline{E}\rightarrow \mathbb{N}_0$, which satisfies the
following \emph{flow conservation constraint} on each vertex $v\in V$:
\begin{align}\label{eq:flow_conservation}
  \sum_{e\in \delta^+_{\overline{E}}(v)} m_v(e^+)f(e) - 
  \sum_{e\in \delta^-_{\overline{E}}(v)} m_v(e^-)f(e) = 0
\end{align}

Note in particular that $f(e^-_v)$ is the input flow for vertex $v$
and $f(e^+_v)$ is its output flow. We will for the remainder of the 
paper refer to integer hyperflows simply as flows.
An example of a flow is shown in Fig. \ref{fig:hyperflow}.%

\begin{figure}[tbp]
  \centering
  \begin{tikzpicture}[abstractNetwork]
    \node[vertex] (B)  at (0,-2)   {$B$} ;
    \node[vertex] (C)  at (4,-2)   {$C$} ;
    \node[vertex] (A)  at (0,-4)  {$A$} ;
    
    \node[vertexBase, overlay] (B1) at (-2,-2) {\phantom{$B$}};
    \node[vertexBase, overlay] (A1) at (-2,-4) {\phantom{$A$}};
    \node[vertexBase, overlay] (C1) at (6,-2) {\phantom{$C$}};
    \node[hyperEdge, label=below:{$e_1$}] (e3) at (2,-2) {$2$};
    \node[hyperEdge, label=below:{$e_2$}] (e) at (4,-4) {$1$};
    
    \draw[arrow] (B) -- (e3);
    \draw[arrow] (e3) -- (C);
    \draw[arrow] (C) -- (e);
    \draw[arrow] (e) -- (A);
    \draw[arrow] (e) -- (B);
    
    \draw[arrow] (B.200) --  node[below] {$1$}(B1.340);
    \draw[arrow] (B1.20) -- node[above] {$2$}(B.160);
    
    \draw[arrow] (A.200) -- node[below] {$1$}(A1.340);
    \draw[arrow] (A1.20) -- node[above] {$0$}(A.160);
    
    \draw[arrow] (C1.200) --node[below] {$0$} (C.340);
    \draw[arrow] (C.20) -- node[above] {$1$}(C1.160);
  \end{tikzpicture}%
  \caption{Example flow $f$ on the extended hypergraph
    from Fig.~\ref{fig:hypergraph}.  Vertex $D$ has been omitted as it
    has no in- or out-flow.  Edges leaving or entering $D$ have also
    been omitted as they have no flow.  The flow on an edge is
    represented by an integer.  For example, the half edge into $B$
    has flow $f(e_B^-)=2$, the half edge leaving $B$ has flow
    $f(e_B^+)=1$, and edge $e_1$ has flow $f(e_1)=2$.}
  \label{fig:hyperflow}%
\end{figure}

\subsection{Petri Nets}
Petri nets are an alternative method to analyse CRNs.  Each molecular species
in the network forms a \emph{place} in the Petri net and each reaction
corresponds to a transition \cite{Koch:2010a,Reddy:93,Reddy:96}. The
stoichiometric matrix commonly used in chemistry has an equivalent in
Petri net terminology, called the incidence matrix~\cite{Koch:2010a}.
In Section~\ref{sec:realisability} we will describe a transformation
of a flow to a Petri net. The following notation for Petri nets
(with the exception of arc weights) follows \cite{esparza:98}.

A \emph{net} is a triple $(P,T,W)$ with a set of places $P$, a set of
transitions $T$, and an arc weight function
$W\colon (P\times T) \cup (T\times P) \rightarrow \mathbb{N}_0$. A
\emph{marking} on a net is a function
$M\colon P\rightarrow \mathbb{N}_0$ assigning a number of tokens to
each place.  With $M_\emptyset$ we denote the empty marking, i.e.,
$M_\emptyset(p) = 0,\: \forall p\in P$.  A \emph{Petri net} is a pair
$(N, M_0)$ of a net $N$ and an initial marking $M_0$. For all
$x \in P \cup T$, we define the \emph{pre-set} as
$^\bullet x = \{y \in P \cup T \mid W(y,x) > 0\}$ and the
\emph{post-set} as $x^\bullet = \{y \in P \cup T \mid W(x,y) >
0\}$. We say that a transition $t$ is enabled by the marking $M$ if
$W(p,t) \leq M(p), \forall p\in P$. When a transition $t$ is enabled
it can \emph{fire}, resulting in a marking $M'$ where
$M'(p)=M(p)-W(p,t)+W(t,p), \: \forall p\in P$. Such a firing is
denoted by $M\xrightarrow{t} M'$. A \emph{firing sequence}~$\sigma$ is
a sequence of firing transitions $\sigma = t_1 t_2 \ldots t_n$. Such a
firing sequence gives rise to a sequence of markings
$M_0 \xrightarrow{t_1} M_1 \xrightarrow{t_2} M_2 \xrightarrow{t_3}
\ldots \xrightarrow{t_n} M_n$ which is denoted by
$M_0 \xrightarrow{\sigma} M_n$. 
In Fig.~\ref{fig:firing_sequence} we present an example of a firing sequence which in this instance is the sequence $\sigma=t_1t_2t_3$.%

\begin{figure}[tbp]
  \centering
  \begin{subfigure}[b]{\textwidth}
  \centering
  \begin{tikzpicture}[abstractNetwork, node distance=2em and 3em,
      vertexPetri/.style={vertex, minimum size=2em},
      hyperEdgePetri/.style={hyperEdge, inner sep=0, minimum width=1.5em, minimum height=1ex},
      token/.style={circle, draw, fill, minimum size=1ex, inner sep=0}
    ]
    \node[vertexPetri, label=left:$p_1$] (p1) {};
    \node[hyperEdgePetri, below right = 1em and 3 em of p1, label = below : $t_1$] (t1) {};
    \node[vertexPetri, below=of p1, label = left :$p_2$] (p2) {};
    \node[vertexPetri, right=of t1, label = below:$p_3$] (p3) {};
    \node[hyperEdgePetri, right=of p3, label=below:$t_2$] (t2) {};
    \node[vertexPetri, right=of t2, label = below : $p_4$] (p4) {};
    \node[hyperEdgePetri, right=of p4, label = below:$t_3$] (t3) {};
    \node[vertexPetri, right=of t3, label = below : $p_5$] (p5) {};
    
    \draw[arrow] (p1) -- (t1);
    \draw[arrow] (p2) -- (t1);
    \draw[arrow] (t1) -- (p3);
    \draw[arrow] (p3) -- (t2);
    \draw[arrow] (t2) -- (p4);
    \draw[arrow] (p4) -- (t3);
    \draw[arrow] (t3.north) [bend right = 10] to (p1);
    \draw[arrow] (t3) -- (p5);

    \node[token] at (p1) {};
    \node[token] at (p2) {};
  \end{tikzpicture}%
  \subcaption{}
  \end{subfigure}
  \begin{subfigure}[b]{\textwidth}
  \centering
  \begin{tikzpicture}[abstractNetwork, node distance=2em and 3em,
      vertexPetri/.style={vertex, minimum size=2em},
      hyperEdgePetri/.style={hyperEdge, inner sep=0, minimum width=1.5em, minimum height=1ex},
      token/.style={circle, draw, fill, minimum size=1ex, inner sep=0}
    ]
    \node[vertexPetri, label=left:$p_1$] (p1) {};
    \node[hyperEdgePetri, below right = 1em and 3 em of p1, label = below : $t_1$] (t1) {};
    \node[vertexPetri, below=of p1, label = left :$p_2$] (p2) {};
    \node[vertexPetri, right=of t1, label = below:$p_3$] (p3) {};
    \node[hyperEdgePetri, right=of p3, label=below:$t_2$] (t2) {};
    \node[vertexPetri, right=of t2, label = below : $p_4$] (p4) {};
    \node[hyperEdgePetri, right=of p4, label = below:$t_3$] (t3) {};
    \node[vertexPetri, right=of t3, label = below : $p_5$] (p5) {};
    
    \draw[arrow] (p1) -- (t1);
    \draw[arrow] (p2) -- (t1);
    \draw[arrow] (t1) -- (p3);
    \draw[arrow] (p3) -- (t2);
    \draw[arrow] (t2) -- (p4);
    \draw[arrow] (p4) -- (t3);
    \draw[arrow] (t3.north) [bend right = 10] to (p1);
    \draw[arrow] (t3) -- (p5);

    \node[token] at (p3) {};
  \end{tikzpicture}%
  \subcaption{}
  \end{subfigure}
  \begin{subfigure}[b]{\textwidth}
  \centering
  \begin{tikzpicture}[abstractNetwork, node distance=2em and 3em,
      vertexPetri/.style={vertex, minimum size=2em},
      hyperEdgePetri/.style={hyperEdge, inner sep=0, minimum width=1.5em, minimum height=1ex},
      token/.style={circle, draw, fill, minimum size=1ex, inner sep=0}
    ]
    \node[vertexPetri, label=left:$p_1$] (p1) {};
    \node[hyperEdgePetri, below right = 1em and 3 em of p1, label = below : $t_1$] (t1) {};
    \node[vertexPetri, below=of p1, label = left :$p_2$] (p2) {};
    \node[vertexPetri, right=of t1, label = below:$p_3$] (p3) {};
    \node[hyperEdgePetri, right=of p3, label=below:$t_2$] (t2) {};
    \node[vertexPetri, right=of t2, label = below : $p_4$] (p4) {};
    \node[hyperEdgePetri, right=of p4, label = below:$t_3$] (t3) {};
    \node[vertexPetri, right=of t3, label = below : $p_5$] (p5) {};
    
    \draw[arrow] (p1) -- (t1);
    \draw[arrow] (p2) -- (t1);
    \draw[arrow] (t1) -- (p3);
    \draw[arrow] (p3) -- (t2);
    \draw[arrow] (t2) -- (p4);
    \draw[arrow] (p4) -- (t3);
    \draw[arrow] (t3.north) [bend right = 10] to (p1);
    \draw[arrow] (t3) -- (p5);

    \node[token] at (p4) {};
  \end{tikzpicture}%
  \subcaption{}
  \end{subfigure}
  \begin{subfigure}[b]{\textwidth}
  \centering
  \begin{tikzpicture}[abstractNetwork, node distance=2em and 3em,
      vertexPetri/.style={vertex, minimum size=2em},
      hyperEdgePetri/.style={hyperEdge, inner sep=0, minimum width=1.5em, minimum height=1ex},
      token/.style={circle, draw, fill, minimum size=1ex, inner sep=0}
    ]
    \node[vertexPetri, label=left:$p_1$] (p1) {};
    \node[hyperEdgePetri, below right = 1em and 3 em of p1, label = below : $t_1$] (t1) {};
    \node[vertexPetri, below=of p1, label = left :$p_2$] (p2) {};
    \node[vertexPetri, right=of t1, label = below:$p_3$] (p3) {};
    \node[hyperEdgePetri, right=of p3, label=below:$t_2$] (t2) {};
    \node[vertexPetri, right=of t2, label = below : $p_4$] (p4) {};
    \node[hyperEdgePetri, right=of p4, label = below:$t_3$] (t3) {};
    \node[vertexPetri, right=of t3, label = below : $p_5$] (p5) {};
    
    \draw[arrow] (p1) -- (t1);
    \draw[arrow] (p2) -- (t1);
    \draw[arrow] (t1) -- (p3);
    \draw[arrow] (p3) -- (t2);
    \draw[arrow] (t2) -- (p4);
    \draw[arrow] (p4) -- (t3);
    \draw[arrow] (t3.north) [bend right = 10] to (p1);
    \draw[arrow] (t3) -- (p5);

    \node[token] at (p5) {};
    \node[token] at (p1) {};
  \end{tikzpicture}%
  \subcaption{}
  \end{subfigure}
  \caption{Example firing sequence. Here $P=\{p_1,p_2,p_3,p_4,p_5\}$, $T=\{t_1,t_2,t_3\}$, $W=\{(p_1,t_1)\mapsto 1,(p_2,t_1) \mapsto 1, (t_1,p_3) \mapsto 1,(p_3,t_2)\mapsto 1,(t_2,p_4) \mapsto 1,(p_4,t_3) \mapsto 1,(t_3,p_5) \mapsto 1,(t_3,p_1)  \mapsto 1\}$, and the initial marking $M_0=\{p_1\mapsto1, p_2\mapsto 1, p_3\mapsto 0, p_4\mapsto 0, p_5\mapsto 0\}$ which is depicted in (a). The firing sequence that leads to (d) is $\sigma=t_1t_2t_3$, which is illustrated through (a) to (d).}
\label{fig:firing_sequence}%
\end{figure}

\section{Realisability of Flows}
\label{sec:realisability}

\citet{flow} described a method (summarized
in~Section~\ref{subsec:flowpaperrecap}) to specify pathways in
CRNs and then proceeded to use ILP to enumerate pathway solutions
fulfilling the specification.
In this paper, we focus on assessing the realisability of such a
pathway solution and on determining a specific order of reactions that
proves its realisability. To this end, we map flows into
Petri nets and rephrase the question of realisability as a particular
reachability question in the resulting Petri net.

\subsection{Flows as Petri Nets}\label{sec:flows_as_nets}
We convert a hypergraph $\mathcal{H} = (V, E)$ to a net
$N = (P, T, W)$ by using the vertices $V$ as the places $P$ and the
edges $E$ as the transitions $T$, and by defining the weight function
from the incidence information as follows: for each vertex/place
$v\in V$ and edge/transition $e = (e^+, e^-) \in E$ let
$W(v, e) = m_v(e^+)$ and $W(e, v) = m_v(e^-)$.  This conversion also
works for extended hypergraphs, where the half-edges result in
transitions with either an empty pre-set or post-set. The transitions
corresponding to input reactions are thus always enabled. 

Given a flow, we would like to constrain the Petri net such that it
yields only 
firing sequences for that particular flow.  We therefore further
convert the extended hypergraph $\overline{\mathcal{H}}$ into an
extended net $(V\cup V_E\cup V_T, \overline{E}, W\cup W_E)$ by adding for each
edge $e\in\overline{E}$ an ``external place'' $v_e\in V_E$ with
connectivity $W(v_e,e)=1$ and for each edge $e^+\in E^+$ adding 
a ``target place''
$v_{e^+} \in V_T$ with connectivity $W(e^+, v_{e^+})=1$.  In the following, we will denote the
extended Petri net again by $N$. We then proceed by translating the
given flow $f$ of $\overline{\mathcal{H}}$ into an initial marking
$M_0$ on the extended net.  To this end, we set $M_0(v)=0$ for
$v\in V\cup V_T$ and $M_0(v_e)=f(e)$ for places $v_e\in V_E$. 
Additionally, 
we set the target marking denoted by $M_T$ to $M_T(v)=0$ for $v\in V\cup V_E$ 
and $M_T(v_{e^+})= f(e^+)$ for places $v_{e^+} \in V_T$.

Transitions in
$(N,M_0)$ therefore can fire at most the number of times specified by
the flow.  Furthermore, any firing sequence
$M_0\xrightarrow{\sigma} M_T$ ending in the target marking must
use each transition exactly the number of times specified by the flow.
As an example, the flow in Fig.~\ref{fig:hyperflow} is converted
to the Petri net in Fig.~\ref{fig:petrinet}.%

\begin{figure}[tbp]
  \centering
  \begin{tikzpicture}[abstractNetwork, node distance=2em and 3em,
      vertexPetri/.style={vertex, minimum size=2em},
      hyperEdgePetri/.style={hyperEdge, inner sep=0, minimum width=1.5em, minimum height=1ex},
      token/.style={circle, draw, fill, minimum size=1ex, inner sep=0}
    ]
    \node[vertexPetri, label=above left:$B$] (B) {};
    \node[hyperEdgePetri, above=of B, label=left:$t_{B_{in}}$] (B1) {};
    \node[vertexPetri, above=of B1, label=above:$p_{B_{in}}$] (B3) {};
    \node[hyperEdgePetri, left=of B, label=below:$t_{B_{out}}$] (B2) {};
    \node[vertexPetri, left=of B2, label=above:$p_{B_{out}}$] (B4) {};
    \node[vertexPetri, above = of B2, label=above:$B_T$] (BG) {};
    \node[hyperEdgePetri, right=of B, label=below:{$e_1$}] (e3) {};
    \node[vertexPetri, above=of e3, label=above:$e_{1_{in}}$] (E3) {};
    \node[vertexPetri, right= of e3, label=above left:$C$] (C) {};
    \node[hyperEdgePetri, above=of C, label=left:$t_{C_{in}}$] (C1) {};
    \node[vertexPetri, above=of C1, label=above:$p_{C_{in}}$] (C3) {};
    \node[hyperEdgePetri, right=of C, label=below:$t_{C_{out}}$] (C2) {};
    \node[vertexPetri, right=of C2, label=above:$p_{C_{out}}$] (C4) {};
    \node[vertexPetri, above = of C2, label=above:$C_T$] (CG) {};
    \node[hyperEdgePetri, below=of C, label=below:{$e_2$}] (e) {};
    \node[vertexPetri, right=of e, label=below:$e_{2_{in}}$] (E1) {};
    \node[vertexPetri, label=below left:$A$] at ($(B |- e)$) (A) {};
    \node[hyperEdgePetri, below=of A, label=below:$t_{A_{in}}$] (A1) {};
    \node[vertexPetri, right=of A1, label=below:$p_{A_{in}}$] (A3) {};
    \node[hyperEdgePetri, left=of A, label=above:$t_{A_{out}}$] (A2) {};
    \node[vertexPetri, left=of A2, label=below:$p_{A_{out}}$] (A4) {};
    \node[vertexPetri, below = of A2, label=below:$A_T$] (AG) {};
    
    \draw[arrow] (B) -- (e3);
    \draw[arrow] (e3) -- (C);
    \draw[arrow] (C) -- (e);
    \draw[arrow] (e) -- (A);
    \draw[arrow] (e) -- (B);
    \draw[arrow] (E3) -- (e3);
    \draw[arrow] (E1) -- (e);
    
    \draw[arrow] (B1) -- (B);
    \draw[arrow] (B) -- (B2);
    
    \draw[arrow] (A1) -- (A);
    \draw[arrow] (A) -- (A2);
    
    \draw[arrow] (C1) -- (C);
    \draw[arrow] (C) -- (C2);
    
    \draw[arrow] (B3) -- (B1);
    \draw[arrow] (B4) -- (B2);
    
    \draw[arrow] (A3) -- (A1);
    \draw[arrow] (A4) -- (A2);
    
    \draw[arrow] (C3) -- (C1);
    \draw[arrow] (C4) -- (C2);
    
    \draw[arrow] (B2) --(BG);
    \draw[arrow] (C2) --(CG);
    \draw[arrow] (A2) --(AG);

    \node[token] at (C4) {};
    \node[token, xshift=-0.7ex] at (E3) {};
    \node[token, xshift=0.7ex] at (E3) {};
    \node[token] at (E1) {};
    \node[token, xshift=0.7ex] at (B3) {};
    \node[token, xshift=-0.7ex] at (B3) {};
    \node[token] at (A4) {};
    \node[token] at (B4) {};
  \end{tikzpicture}%
  \caption{The flow from Fig.~\ref{fig:hyperflow}
    converted to a Petri net with its initial marking. Places are circles, transitions are
    rectangles, and tokens are black dots. Arrows indicate pairs of
    places and transitions for which the weight function~$W$ is
    non-zero (in this example, all non-zero weights are equal to
    one). The target marking is $M_T(A_T)=1$, $M_T(B_T)=1$, $M_T(C_T)=1$ and $M_T(p)=0$ for all $p\in P \setminus \{A_T, B_T, C_T\}$.We have omitted the part of the net that corresponds to the
    omitted part of Fig.~\ref{fig:hyperflow}.}
  \label{fig:petrinet}%
\end{figure}

\subsection{Realisability of Flows}
We are interested in whether a given pathway, represented by a flow $f$ on
an extended hypergraph $\overline{\mathcal{H}} = (V,\overline{E})$, is
realisable in the following sense:
Given the input molecules specified by the input flow,
is there a sequence of reactions that respects the flow,
which in the end produces the output molecules specified by the output flow?
In the light of the construction of $(N,M_0)$
  from $(\overline{\mathcal{H}},f)$, this question translates into a
  reachability problem on a Petri net.
\begin{definition} 
  A flow $f$ on $\overline{\mathcal{H}}$ is realisable if there is a firing
  sequence $M_0\xrightarrow{*} M_T$ on the Petri net $(N,M_0)$
  constructed from $(\overline{\mathcal{H}},f)$.
  \label{def:realisability}
\end{definition}

\begin{figure}[tbp]
  \centering
  \begin{tikzpicture}[abstractNetwork]
    \node[vertex] (A)  at (0,0)   {$A$} ;
    \node[vertex] (B)  at (4,0)   {$B$} ;
    \node[vertex] (C)  at (0,2)  {$C$} ;
    \node[vertex] (D)  at (4,2)  {$D$} ;
    
    \node[vertexBase, overlay] (A1) at (-2,0) {\phantom{$A$}};
    \node[vertexBase, overlay] (B1) at (6,0) {\phantom{$B$}};
    \node[vertexBase, overlay] (C1) at (-2,2) {\phantom{$C$}};
    \node[vertexBase, overlay] (D1) at (6,2) {\phantom{$D$}};
    \node[hyperEdge, label=below:{$e_1$}] (e1) at (2,0) {$1$};
    \node[hyperEdge, label=below:{$e_2$}] (e2) at (2,4) {$1$};
    
    \draw[arrow] (A) -- (e1);
    \draw[arrow] (e1) -- (B);
    \draw[arrow] (e2) to [bend right=30] (C);
    \draw[arrow] (C) to [bend right=30] (e1);
    \draw[arrow] (D) to [bend right=30] (e2);
    \draw[arrow] (e1) to [bend right=30] (D);		

    \draw[arrow] (B.20) -- node[above] {$1$}  (B1.160);
    \draw[arrow] (B1.200) -- node[below] {$0$} (B.340);
    
    \draw[arrow] (A.200) -- node[below] {$0$} (A1.340);
    \draw[arrow] (A1.20) -- node[above] {$1$} (A.160);
    
    \draw[arrow] (D.20) -- node[above] {$0$} (D1.160);
    \draw[arrow] (D1.200)-- node[below] {$0$} (D.340);
    
    \draw[arrow] (C1.20) -- node[above] {$0$} (C.160);
    \draw[arrow] (C.200)-- node[below] {$0$} (C1.340);
  \end{tikzpicture}%
  \caption{Example of a flow which is not realisable.
    Observe that the flow is indeed viable as it fulfils the
    flow conservation constraint. Furthermore, notice that there is no
    input flow to neither $C$ nor $D$, and therefore in the
    corresponding Petri net it will not be possible to fire either of
    $e_1$ or $e_2$ which is necessary for it to be realised. However,
    if $C$ or $D$ was borrowed the related flow with this borrowing
    would be realisable.}
  \label{fig:missing_catalyst}%
\end{figure}

Fig.~\ref{fig:missing_catalyst} shows that not all flows $f$ on $\overline{\mathcal{H}}$ are realisable. In this example it is impossible
to realise the flow as long as there is no flow entering either $C$ or $D$.
For the flow in Fig. \ref{fig:hyperflow}, on the other hand, such a firing
sequence exists.  The firing sequences corresponding to a realisable flow
are not unique in general.  For instance, the Petri net constructed from
the flow presented in Fig.~\ref{fig:petrinet} can reach the
target marking $M_T$ in essentially two different manners.  Modulo
the firing of input/output transitions, those two firing subsequences are
$e_1e_1e_2$ and $e_1e_2e_1$.
For a chemical example of a realisable flow see Fig. \ref{fig:realisable_flow_formose}. This is a flow for the formose reaction.%

\begin{figure}[tbp]
\centering
\resizebox{\textwidth}{!}{%
  \renewcommand\modInputPrefix{data/formose_realisable_flow}%
  \input{data/formose_realisable_flow/out/018_dg_0_01100_f_0_0_filt.tex}%
  }%
\caption{An example of a flow for the formose reaction which is realisable. The input compound Formald is marked with green and Glycoald which is both an input and output compound is marked with turquoise.}
\label{fig:realisable_flow_formose}
\end{figure}

\subsection{Realisability Certificate}
In order to introduce realisability certificates that describe
the causal order of the reactions needed to make the pathway realisable,
we need some established terminology.
\begin{definition}[Occurrence Net \cite{goltz:83}]\label{def:occ_net}
  A net $K=(P_K,T_K,F_K)$ with $F_K \subseteq (P_K \times T_K) \cup (T_K \times P_K)$ is an
  occurrence net iff 
  \begin{enumerate}
  \item $\forall x, y \in P_K \cup T_K \: x F^+_K y \Rightarrow \neg (y F^+_K x)$ ($F^+_K$ denoting the transitive closure of $F_K$);
  \item $\forall p \in P_K \: |^\bullet p| \leq 1 \land |p^\bullet| \leq 1$.
  \end{enumerate}
\end{definition}
``Occurrence net'' is also defined in \cite{genrich:80,best:82}, but is used with a different meaning in other sources, see e.g. \cite{nielsen:81}.
\begin{definition}[Process \cite{goltz:83} (adapted)]
  Let $N=(P_N, T_N, W_N, M_0)$ be a Petri net and $M$ a reachable marking in
  $N$.  A \emph{process} is a pair $(K, q)$ of an occurrence net
  $K=(P_K,T_K,F_K)$ and a mapping $q\colon K\rightarrow N$ which starts in $M$ and satisfies the
  following properties
  \begin{enumerate}
  \item $q(P_K)  \subseteq P_N$ and $q(T_K) \subseteq T_N$;
  \item If $C \mathrel{:=} \{x \in P_K \mid {^\bullet} x = \emptyset\}$ then
    $M(p) = |q^{-1}(p) \cap C|$ for all $p \in P_N$;
  \item $W_N(p,q(t)) = |q^{-1}(p)\cap ^\bullet t|$ and $W_N(q(t),p) =
    |q^{-1}(p)\cap t^\bullet|$ for all $t\in T_K$ and $p \in P_N$.
  \end{enumerate}
\end{definition}
A process is thus an occurrence net that maps back to a Petri net, such that
it respects the transitions, places and weight function of the Petri net.
Furthermore, the process starts at the marking $M$ in the net.
\begin{definition}
  A \emph{realisability certificate} for $(\overline{\mathcal{H}},f)$ is a
  process for the Petri net $(N, M_0)$ constructed from
  $(\overline{\mathcal{H}},f)$ that leads from the initial marking $M_0$ to
  the target marking $M_T$.
\end{definition}

A realisability certificate exists if and only if the target marking
$M_T$ is reachable from the initial marking $M_0$ \cite[Theorem
3.6]{goltz:83}.

A realisability certificate can be constructed from the initial marking using an algorithm exemplified in \cite{goltz:83}. Furthermore, the Petri net tool \emph{A Low Level Analyzer} (LoLA)\cite{Schmidt:2000} is, given a Petri net with its initial marking and a target marking, able to compute a so-called witness path, which is an object isomorphic to a realisability certificate (or tell if the target marking is unreachable and no realisability certificate exists). The computational complexity of reachability in Petri nets is a complex question in the general case \cite{mayr:81, reutenauer:90}. However, in practical cases LoLA performs well---in particular, in our use cases it normally finishes in less than 10 minutes. In this paper, we used LoLA to produce the underlying certificate for our figures. 
For an example of a realisability certificate see Fig.~\ref{fig:cert_example_flow},
which is a certificate for the flow in Fig.~\ref{fig:hyperflow}. 
For a more chemical example of a 
realisabiltiy certificate see Fig.~\ref{fig:realisable_DAG_formose}, which is for the flow in
Fig.~\ref{fig:realisable_flow_formose}. 
To draw realisability certificates more
concisely we have omitted $q^{-1}(v)$ for all $v\in V_E\cup V_T$, i.e., the places in
the occurrence net that correspond to the external or target places in the Petri net,
as well as $q^{-1}(v_e, e)$ for all  $v_e \in V_E$ and $q^{-1}(e^+, v_{e^+})$ for all $v_{e^+} \in V_T$, 
i.e., the arcs leaving the external places or entering the target places in the Petri net. We have also omitted transitions
on which the corresponding edges have no flow and places corresponding to
vertices with no in- nor out-flow.

\begin{figure}[tbp]
  \centering
  \begin{tikzpicture}[abstractNetwork]
    \node[vertex] (B1)  at (0,0)   {$B$} ;
    \node[hyperEdge, label=left:${e_1}$, below = of B1]  (e11){};
    \node[vertex, below = of e11] (C1) {$C$};
    
    \draw[arrow] (B1) -- (e11);
    \draw[arrow] (e11) -- (C1);
    
    \node[vertex, right = of B1] (B2) {$B$} ;
    \node[hyperEdge, label=left:${e_1}$, below = of B2]  (e12){};
    \node[vertex, below = of e12] (C2) {$C$};
    \node[hyperEdge, label=left:${e_2}$, below = of C2]  (e2){};
    \node[vertex, below left = of e2] (A) {$A$};
    \node[vertex, below right = of e2] (B3) {$B$};
      
    \draw[arrow] (B2) -- (e12);
    \draw[arrow] (e12) -- (C2);
    \draw[arrow] (C2) -- (e2);
    \draw[arrow] (e2) -- (A);
    \draw[arrow] (e2) -- (B3);
  \end{tikzpicture}%
  \caption{A realisability certificate for the flow in Fig.~\ref{fig:hyperflow}.}
  \label{fig:cert_example_flow}
\end{figure}

\begin{figure}[tbp]
\centering
\resizebox{\textwidth}{!}{%
  \renewcommand\modInputPrefix{data/formose_realisable_cert}%
  \input{data/formose_realisable_cert/out/019_dg_0_01100.tex}%
  }%
\caption{A realisability certificate for the flow in Fig. \ref{fig:realisable_flow_formose}. The input compounds are marked with green and the output compounds are marked with blue.}
\label{fig:realisable_DAG_formose}
\end{figure}

A realisability certificate is a directed acyclic graph (DAG) by Def.~\ref{def:occ_net} (1), hence it has a topological sorting \cite{cormen:09}, i.e., a linear ordering of the vertices such that for every edge $(u,v)$, $u$ comes before $v$ in the ordering. Such a topological sorting of the realisability certificates produces one possible firing sequence of its transitions which realises the flow.

Finally, we note that a realisability certificate is formulated in the Petri net literature such that it gives an individual token
interpretation, where individual tokens are distinguishable
\cite{glabbeek:05}. Such a property is an advantage (actually, a necessity) if one is to do atom tracing~\cite{trace} of stable isotope atoms through the pathway.

\section{Extended Realisability}\label{sec:extended_real}

We have demonstrated above that flows may not be realisable. In this
section, we study various means by which non-realisable flows may be
made realisable.
\begin{definition}[Scaled-Realisable]
  A flow $f$ on an extended hypergraph
  $\overline{\mathcal{H}}=(V,\overline{E})$ is \emph{scaled-realisable}, if
  there exists an integer $k\ge 1$ such that the resulting flow $k\cdot f$ is realisable.
\end{definition}

Asking if a flow $f$ is scaled-realisable corresponds to
asking if $k$ copies of $f$ can be realised concurrently. This is of
interest as in the real world, a pathway is often not just happening once,
but multiple times. Therefore, even if the flow is not
realisable, it is meaningful to consider if the scaled flow is. 
Fig. \ref{fig:scaled_realisable_flow_formose_id} is an example of such a flow which is not realisable itself, but is scaled-realisable by a factor $2$. The flow represents an alternative formose reaction. In order to see that this flow is indeed scaled-realisable, see the realisability certificate of the flow in Fig. \ref{fig:scaled_realisable_DAG_formose_id}.%

\begin{figure}[tbp]
\centering
\resizebox{\textwidth}{!}{%
  \renewcommand\modInputPrefix{data/formose_scaled_flow_id}%
  % !TEX root = ../../../paper_lncs.tex
\begin{tikzpicture}[abstractNetwork, node distance=2em and 2em]
% id = 0, graphName = Glycoald
\node[modStyleDGHyperVertex, inoutMol, text depth=0] (v-0-0) {{$\mathrm{Glycoald}$}};
\node[modStyleDGHyperVertexHidden, above = of v-0-0] (above-v-0-0) {};
% id = 3{ 'Glycoald' }, 'Keto-enol isomerization ->', { 'EtD' }
\node[modStyleDGHyperEdge, xshift=-2em, right = of above-v-0-0] (v-3-0) {$\mathrm{1}$};
% id = 4{ 'EtD' }, 'Keto-enol isomerization <-', { 'Glycoald' }
\node[modStyleDGHyperEdge, xshift=2em, left = of above-v-0-0] (v-4-0) {$\mathrm{1}$};
% id = 2, graphName = EtD
\node[modStyleDGHyperVertex, above = of above-v-0-0] (v-2-0) {{$\mathrm{EtD}$}};
% id = 6{ 'Formald' 'EtD' }, 'Aldol Addition ->', { 'Glyald' }
\node[modStyleDGHyperEdge, above = of v-2-0] (v-6-0) {$\mathrm{2}$};
% id = 1, graphName = Formald
\node[modStyleDGHyperVertex, above = of v-6-0, inputMol] (v-1-0) {{$\mathrm{Formald}$}};
% id = 5, graphName = Glyald
\node[modStyleDGHyperVertex, text depth=0, right = of v-6-0] (v-5-0) {{$\mathrm{Glyald}$}};
% id = 10{ 'Glyald' }, 'Keto-enol isomerization ->', { 'Propenetriol' }
\node[modStyleDGHyperEdge, above = of v-5-0] (v-10-0) {$\mathrm{1}$};
% id = 9, graphName = Propenetriol
\node[modStyleDGHyperVertex, right = of v-10-0] (v-9-0) {{$\mathrm{Propenetriol}$}};
% id = 36{ 'Glyald' 'Propenetriol' }, 'Aldol Addition ->', { '2Hex' }
\node[modStyleDGHyperEdge, below = of v-9-0] (v-36-0) {$\mathrm{1}$};
% id = 35, graphName = 2Hex
\node[modStyleDGHyperVertex, right = of v-36-0] (v-35-0) {{$\mathrm{2Hex}$}};
% id = 53{ '2Hex' }, 'Keto-enol isomerization ->', { 'EnolHex' }
\node[modStyleDGHyperEdge, right = of v-35-0] (v-53-0) {$\mathrm{1}$};
% id = 21, graphName = EnolHex
\node[modStyleDGHyperVertex, right = of v-53-0] (v-21-0) {{$\mathrm{EnolHex}$}};
% id = 65{ 'EnolHex' }, 'Keto-enol isomerization <-', { 'Aldohex' }
\node[modStyleDGHyperEdge, below=of v-21-0] (v-65-0) {$\mathrm{1}$};
\node[modStyleDGHyperVertex, below = of v-65-0] (v-15-0) {{$\mathrm{Aldohex}$}};
% id = 48{ 'Aldohex' }, 'Aldol Addition <-', { 'EtD' 'E4' }
\node[modStyleDGHyperEdge, yshift=-0.5mm, below = of v-15-0] (v-48-0) {$\mathrm{1}$};
% id = 7, graphName = E4
\node[modStyleDGHyperVertex, text depth=0, left = of v-48-0] (v-7-0) {{$\mathrm{E4}$}};
% id = 18{ 'E4' }, 'Aldol Addition <-', { 'Glycoald' 'EtD' }
\node[modStyleDGHyperEdge, left = of v-7-0] (v-18-0) {$\mathrm{1}$};
% id = 3{ 'Glycoald' }, 'Keto-enol isomerization ->', { 'EtD' }
\path[modStyleDGHyperConnector] (v-0-0) to (v-3-0);
\path[modStyleDGHyperConnector] (v-3-0) to (v-2-0);
% id = 4{ 'EtD' }, 'Keto-enol isomerization <-', { 'Glycoald' }
\path[modStyleDGHyperConnector] (v-2-0) to (v-4-0);
\path[modStyleDGHyperConnector] (v-4-0) to (v-0-0);
% id = 6{ 'Formald' 'EtD' }, 'Aldol Addition ->', { 'Glyald' }
\path[modStyleDGHyperConnector] (v-1-0) to (v-6-0);
\path[modStyleDGHyperConnector] (v-2-0) to (v-6-0);
\path[modStyleDGHyperConnector] (v-6-0) to (v-5-0);
\path[modStyleDGHyperConnector] (v-5-0) to (v-10-0);
\path[modStyleDGHyperConnector] (v-10-0) to (v-9-0);
\path[modStyleDGHyperConnector] (v-7-0) to (v-18-0);
\path[modStyleDGHyperConnector] (v-18-0) to (v-0-0);
\path[modStyleDGHyperConnector] (v-18-0) to (v-2-0);
\path[modStyleDGHyperConnector] (v-5-0) to (v-36-0);
\path[modStyleDGHyperConnector] (v-9-0) to (v-36-0);
\path[modStyleDGHyperConnector] (v-36-0) to (v-35-0);
\path[modStyleDGHyperConnector] (v-15-0) to (v-48-0);
\path[modStyleDGHyperConnector] (v-48-0) to (v-2-0);
\path[modStyleDGHyperConnector] (v-48-0) to (v-7-0);
\path[modStyleDGHyperConnector] (v-35-0) to (v-53-0);
\path[modStyleDGHyperConnector] (v-53-0) to (v-21-0);
\path[modStyleDGHyperConnector] (v-21-0) to (v-65-0);
\path[modStyleDGHyperConnector] (v-65-0) to (v-15-0);

% inFlow/outFlow, id = 0-0, graphName = Glycoald, inFlow = 1, outFlow = 2
\node[modStyleDGHyperVertexHidden, left = of v-0-0] (v-0-0-IOFlow) {};
\path[modStyleDGHyperConnector] (v-0-0-IOFlow) to[modStyleDGHyperHasReverseShortcut] node[auto, swap] {$\mathrm{1}$} (v-0-0);
\path[modStyleDGHyperConnector] (v-0-0) to[modStyleDGHyperHasReverseShortcut] node[auto, swap] {$\mathrm{2}$} (v-0-0-IOFlow);
% inFlow/outFlow, id = 1-0, graphName = Formald, inFlow = 2, outFlow = 0
\node[modStyleDGHyperVertexHidden, left = of v-1-0] (v-1-0-IOFlow) {};
\path[modStyleDGHyperConnector] (v-1-0-IOFlow) to node[auto, swap] {$\mathrm{2}$} (v-1-0);
\end{tikzpicture}%
%
  }%
\caption{An example of a flow for the formose reaction which is not realisable but is scaled-realisable by a factor $2$. The input compound Formald is marked with green and Glycoald which is both an input and output compound is marked with turquoise. The SMILES strings for all molecule identifiers are listed in Appendix, Table~\ref{tab:smiles}.}
\label{fig:scaled_realisable_flow_formose_id}
\end{figure}

\begin{figure}[tbp]
\centering
\resizebox{\textwidth}{!}{%
  \renewcommand\modInputPrefix{data/formose_scaled_cert_id}%
  % !TEX root = ../../../paper_lncs.tex
\begin{tikzpicture}[abstractNetwork, node distance=2em and 2.5em]
% id = 0, graphName = Glycoald
\node[modStyleDGHyperVertex, inputMol] (v-0-8) {{$\mathrm{Glycoald}$}};
% id = 3{ 'Glycoald' }, 'Keto-enol isomerization ->', { 'EtD' }
\node[modStyleDGHyperEdge,right = of v-0-8] (v-3-13) {};
% id = 2, graphName = EtD
\node[modStyleDGHyperVertex, right = of v-3-13] (v-2-15) {{$\mathrm{EtD}$}};
% id = 6{ 'Formald' 'EtD' }, 'Aldol Addition ->', { 'Glyald' }
\node[modStyleDGHyperEdge, right = of v-2-15] (v-6-17) {};
% id = 5, graphName = Glyald
\node[modStyleDGHyperVertex, right = of v-6-17] (v-5-19) {{$\mathrm{Glyald}$}};
% id = 1, graphName = Formald
\node[modStyleDGHyperVertex, below = of v-6-17, inputMol] (v-1-9) {{$\mathrm{Formald}$}};
% id = 10{ 'Glyald' }, 'Keto-enol isomerization ->', { 'Propenetriol' }
\node[modStyleDGHyperEdge, right = of v-5-19] (v-10-20) {};
% id = 9, graphName = Propenetriol
\node[modStyleDGHyperVertex,right = of v-10-20] (v-9-21) {{$\mathrm{Propenetriol}$}};
% id = 36{ 'Glyald' 'Propenetriol' }, 'Aldol Addition ->', { '2Hex' }
\node[modStyleDGHyperEdge, right = of v-9-21] (v-36-22) {};
% id = 5, graphName = Glyald
\node[modStyleDGHyperVertex, below = of v-36-22] (v-5-18) {{$\mathrm{Glyald}$}};
% id = 6{ 'Formald' 'EtD' }, 'Aldol Addition ->', { 'Glyald' }
\node[modStyleDGHyperEdge, left = of v-5-18] (v-6-16) {};
% id = 1, graphName = Formald
\node[modStyleDGHyperVertex, left = of v-6-16, inputMol] (v-1-10) {{$\mathrm{Formald}$}};
% id = 2, graphName = EtD
\node[modStyleDGHyperVertex, below = of v-6-16] (v-2-14) {{$\mathrm{EtD}$}};
% id = 3{ 'Glycoald' }, 'Keto-enol isomerization ->', { 'EtD' }
\node[modStyleDGHyperEdge,left = of v-2-14] (v-3-12) {};
% id = 0, graphName = Glycoald
\node[modStyleDGHyperVertex, left = of v-3-12, inputMol] (v-0-7) {{$\mathrm{Glycoald}$}};
% id = 35, graphName = 2Hex
\node[modStyleDGHyperVertex, right = of v-36-22] (v-35-23) {{$\mathrm{2Hex}$}};
% id = 53{ '2Hex' }, 'Keto-enol isomerization ->', { 'EnolHex' }
\node[modStyleDGHyperEdge, right = of v-35-23] (v-53-24) {};
% id = 21, graphName = EnolHex
\node[modStyleDGHyperVertex, below = of v-53-24] (v-21-25) {{$\mathrm{EnolHex}$}};
% id = 65{ 'EnolHex' }, 'Keto-enol isomerization <-', { 'Aldohex' }
\node[modStyleDGHyperEdge, below = of v-21-25] (v-65-26) {};
% id = 15, graphName = Aldohex
\node[modStyleDGHyperVertex, below = of v-65-26] (v-15-27) {{$\mathrm{Aldohex}$}};
% id = 48{ 'Aldohex' }, 'Aldol Addition <-', { 'EtD' 'E4' }
\node[modStyleDGHyperEdge, left = of v-15-27] (v-48-28) {};
% id = 7, graphName = E4
\node[modStyleDGHyperVertex, left = of v-48-28] (v-7-29) {{$\mathrm{E4}$}};
% id = 18{ 'E4' }, 'Aldol Addition <-', { 'Glycoald' 'EtD' }
\node[modStyleDGHyperEdge, left = of v-7-29] (v-18-31) {};
% id = 0, graphName = Glycoald
\node[modStyleDGHyperVertex,below = of v-18-31, outputMol] (v-0-34) {{$\mathrm{Glycoald}$}};
% id = 2, graphName = EtD
\node[modStyleDGHyperVertex, left = of v-18-31] (v-2-33) {{$\mathrm{EtD}$}};
% id = 6{ 'Formald' 'EtD' }, 'Aldol Addition ->', { 'Glyald' }
\node[modStyleDGHyperEdge, left = of v-2-33] (v-6-36) {};
% id = 1, graphName = Formald
\node[modStyleDGHyperVertex, below = of v-6-36, inputMol] (v-1-6) {{$\mathrm{Formald}$}};
% id = 5, graphName = Glyald
\node[modStyleDGHyperVertex,left = of v-6-36] (v-5-39) {{$\mathrm{Glyald}$}};
% id = 2, graphName = EtD
\node[modStyleDGHyperVertex, below = of v-7-29] (v-2-30) {{$\mathrm{EtD}$}};
% id = 6{ 'Formald' 'EtD' }, 'Aldol Addition ->', { 'Glyald' }
\node[modStyleDGHyperEdge, below = of v-2-30] (v-6-32) {};
% id = 1, graphName = Formald
\node[modStyleDGHyperVertex, right = of v-6-32, inputMol] (v-1-11) {{$\mathrm{Formald}$}};
% id = 5, graphName = Glyald
\node[modStyleDGHyperVertex, left = of v-6-32] (v-5-35) {{$\mathrm{Glyald}$}};
% id = 10{ 'Glyald' }, 'Keto-enol isomerization ->', { 'Propenetriol' }
\node[modStyleDGHyperEdge, left = of v-5-35] (v-10-38) {};
% id = 9, graphName = Propenetriol
\node[modStyleDGHyperVertex, left = of v-10-38] (v-9-40) {{$\mathrm{Propenetriol}$}};
% id = 36{ 'Glyald' 'Propenetriol' }, 'Aldol Addition ->', { '2Hex' }
\node[modStyleDGHyperEdge, left = of v-9-40] (v-36-41) {};
% id = 35, graphName = 2Hex
\node[modStyleDGHyperVertex, left = of v-36-41] (v-35-42) {{$\mathrm{2Hex}$}};
% id = 53{ '2Hex' }, 'Keto-enol isomerization ->', { 'EnolHex' }
\node[modStyleDGHyperEdge, left = of v-35-42] (v-53-43) {};
% id = 21, graphName = EnolHex
\node[modStyleDGHyperVertex, left=of v-53-43] (v-21-44) {{$\mathrm{EnolHex}$}};
% id = 65{ 'EnolHex' }, 'Keto-enol isomerization <-', { 'Aldohex' }
\node[modStyleDGHyperEdge, below = of v-21-44] (v-65-45) {};
% id = 15, graphName = Aldohex
\node[modStyleDGHyperVertex, right = of v-65-45] (v-15-46) {{$\mathrm{Aldohex}$}};
% id = 48{ 'Aldohex' }, 'Aldol Addition <-', { 'EtD' 'E4' }
\node[modStyleDGHyperEdge, right = of v-15-46] (v-48-47) {};
% id = 2, graphName = EtD
\node[modStyleDGHyperVertex,below = of v-48-47] (v-2-49) {{$\mathrm{EtD}$}};
% id = 4{ 'EtD' }, 'Keto-enol isomerization <-', { 'Glycoald' }
\node[modStyleDGHyperEdge, left = of v-2-49] (v-4-51) {};
% id = 0, graphName = Glycoald
\node[modStyleDGHyperVertex, left = of v-4-51, outputMol] (v-0-54) {{$\mathrm{Glycoald}$}};
% id = 7, graphName = E4
\node[modStyleDGHyperVertex, right= of v-48-47] (v-7-48) {{$\mathrm{E4}$}};
% id = 18{ 'E4' }, 'Aldol Addition <-', { 'Glycoald' 'EtD' }
\node[modStyleDGHyperEdge, right = of v-7-48] (v-18-50) {};
% id = 0, graphName = Glycoald
\node[modStyleDGHyperVertex, below = of v-18-50, outputMol] (v-0-53) {{$\mathrm{Glycoald}$}};
% id = 2, graphName = EtD
\node[modStyleDGHyperVertex, right = of v-18-50] (v-2-52) {{$\mathrm{EtD}$}};
% id = 4{ 'EtD' }, 'Keto-enol isomerization <-', { 'Glycoald' }
\node[modStyleDGHyperEdge, right = of v-2-52] (v-4-55) {};
% id = 0, graphName = Glycoald
\node[modStyleDGHyperVertex, right = of v-4-55, outputMol] (v-0-58) {{$\mathrm{Glycoald}$}};

% id = 3{ 'Glycoald' }, 'Keto-enol isomerization ->', { 'EtD' }
\path[modStyleDGHyperConnector] (v-0-7) to (v-3-12);
\path[modStyleDGHyperConnector] (v-3-12) to (v-2-14);
\path[modStyleDGHyperConnector] (v-0-8) to (v-3-13);
\path[modStyleDGHyperConnector] (v-3-13) to (v-2-15);
% id = 4{ 'EtD' }, 'Keto-enol isomerization <-', { 'Glycoald' }
\path[modStyleDGHyperConnector] (v-2-49) to (v-4-51);
\path[modStyleDGHyperConnector] (v-4-51) to (v-0-54);
\path[modStyleDGHyperConnector] (v-2-52) to (v-4-55);
\path[modStyleDGHyperConnector] (v-4-55) to (v-0-58);
% id = 6{ 'Formald' 'EtD' }, 'Aldol Addition ->', { 'Glyald' }
\path[modStyleDGHyperConnector] (v-1-10) to (v-6-16);
\path[modStyleDGHyperConnector] (v-2-14) to (v-6-16);
\path[modStyleDGHyperConnector] (v-6-16) to (v-5-18);
\path[modStyleDGHyperConnector] (v-1-9) to (v-6-17);
\path[modStyleDGHyperConnector] (v-2-15) to (v-6-17);
\path[modStyleDGHyperConnector] (v-6-17) to (v-5-19);
\path[modStyleDGHyperConnector] (v-1-11) to (v-6-32);
\path[modStyleDGHyperConnector] (v-2-30) to (v-6-32);
\path[modStyleDGHyperConnector] (v-6-32) to (v-5-35);
\path[modStyleDGHyperConnector] (v-1-6) to (v-6-36);
\path[modStyleDGHyperConnector] (v-2-33) to (v-6-36);
\path[modStyleDGHyperConnector] (v-6-36) to (v-5-39);
% id = 10{ 'Glyald' }, 'Keto-enol isomerization ->', { 'Propenetriol' }
\path[modStyleDGHyperConnector] (v-5-19) to (v-10-20);
\path[modStyleDGHyperConnector] (v-10-20) to (v-9-21);
\path[modStyleDGHyperConnector] (v-5-35) to (v-10-38);
\path[modStyleDGHyperConnector] (v-10-38) to (v-9-40);
% id = 18{ 'E4' }, 'Aldol Addition <-', { 'Glycoald' 'EtD' }
\path[modStyleDGHyperConnector] (v-7-29) to (v-18-31);
\path[modStyleDGHyperConnector] (v-18-31) to (v-0-34);
\path[modStyleDGHyperConnector] (v-18-31) to (v-2-33);
\path[modStyleDGHyperConnector] (v-7-48) to (v-18-50);
\path[modStyleDGHyperConnector] (v-18-50) to (v-0-53);
\path[modStyleDGHyperConnector] (v-18-50) to (v-2-52);
% id = 36{ 'Glyald' 'Propenetriol' }, 'Aldol Addition ->', { '2Hex' }
\path[modStyleDGHyperConnector] (v-5-18) to (v-36-22);
\path[modStyleDGHyperConnector] (v-9-21) to (v-36-22);
\path[modStyleDGHyperConnector] (v-36-22) to (v-35-23);
\path[modStyleDGHyperConnector] (v-5-39) to (v-36-41);
\path[modStyleDGHyperConnector] (v-9-40) to (v-36-41);
\path[modStyleDGHyperConnector] (v-36-41) to (v-35-42);
% id = 48{ 'Aldohex' }, 'Aldol Addition <-', { 'EtD' 'E4' }
\path[modStyleDGHyperConnector] (v-15-27) to (v-48-28);
\path[modStyleDGHyperConnector] (v-48-28) to (v-2-30);
\path[modStyleDGHyperConnector] (v-48-28) to (v-7-29);
\path[modStyleDGHyperConnector] (v-15-46) to (v-48-47);
\path[modStyleDGHyperConnector] (v-48-47) to (v-2-49);
\path[modStyleDGHyperConnector] (v-48-47) to (v-7-48);
% id = 53{ '2Hex' }, 'Keto-enol isomerization ->', { 'EnolHex' }
\path[modStyleDGHyperConnector] (v-35-23) to (v-53-24);
\path[modStyleDGHyperConnector] (v-53-24) to (v-21-25);
\path[modStyleDGHyperConnector] (v-35-42) to (v-53-43);
\path[modStyleDGHyperConnector] (v-53-43) to (v-21-44);
% id = 65{ 'EnolHex' }, 'Keto-enol isomerization <-', { 'Aldohex' }
\path[modStyleDGHyperConnector] (v-21-25) to (v-65-26);
\path[modStyleDGHyperConnector] (v-65-26) to (v-15-27);
\path[modStyleDGHyperConnector] (v-21-44) to (v-65-45);
\path[modStyleDGHyperConnector] (v-65-45) to (v-15-46);
\end{tikzpicture}%
%
  }%
\caption{A realisability certificate for the flow in Fig. \ref{fig:scaled_realisable_flow_formose_id} when scaled by a factor $2$, making it scaled-realisable. The input compounds are marked with green and the output compounds are marked with blue. The SMILES strings for all molecule identifiers are listed in Appendix, Table~\ref{tab:smiles}.}
\label{fig:scaled_realisable_DAG_formose_id}
\end{figure}

However, not all flows are scaled-realisable. A counter-example is the
flow presented in Fig.~\ref{fig:missing_catalyst}: no integer scaling
can alleviate the fact that firing $e_1$ or $e_2$ requires $C$ or $D$
to be present at the outset. We note that Thm.~\ref{thm:non-scaled}
from Sec.~\ref{sec:math_realisability} provides an easily checkable
condition which if true implies that a flow is not scaled-realisable.

\begin{definition}[Borrow-Realisable] 
  Let $f$ be a flow on an extended hypergraph,
  $\overline{\mathcal{H}}=(V,\overline{E})$ and let $b$ be a function
  $b\colon V\rightarrow \mathbb{N}_0$.  Set
  $f'(e^-_v) = b(v) + f(e^-_v)$ and $f'(e^+_v) = b(v) + f(e^+_v)$ for
  all $v\in V$, and $f'(e) = f(e)$ for all $e\in E$. Then $f$ is
  \emph{borrow-realisable} if there exists a function $b$
  such that $f'$ is realisable.
\end{definition}
We denote $b$ as the \emph{borrowing function} and we say that $f'$ is
the flow $f$ where $v\in V$ has been borrowed $b(v)$ times. This
models that molecules required for reactions in the pathway can be
aquired from the environment (and returned afterwards). Formally, this
is specified by having an additional input and output flow $b(v)$ for
species $v$.  Furthermore, for a borrowing function $b$ we define
$|b| = \sum_{v\in V} b(v)$, i.e., the total count of molecules
borrowed. The idea of borrowing tokens in the corresponding Petri net
setting has been proposed in \cite[Proposition 10]{desel:1998}
together with a theorem which implies that $f'$ is realisable for some
$b$ with sufficiently large $|b|$. That is, every flow is in fact
borrow-realisable.

The combinatorics underlying the non-oxidative phase of the PPP has
been studied in a series of works focusing on simplifying principles
that explain the structure of metabolic networks, see e.g.\
\cite{Noor2010,Melendez-Hevia:1985}.  An example of a simple flow from
the PPP that is not scaled-realisable is shown in Fig.\
\ref{fig:borrow_realisable_flow_PPP}.  Here, the production of
glyceraldehyde (Glyald) is dependent of the presence of Hex-2-ulose
(2Hex), which depends on fructose 1-phosphate (F1P), which in turn
depends on Glyald.  This cycle of dependencies by
Thm.~\ref{thm:non-scaled} implies that firing is impossible unless one
of the molecules in this cycle is present at the outset, which cannot
be achieved by scaling.  As illustrated in
Fig.~\ref{fig:borrow_realisable_DAG_PPP} and proven by the existence
of the realisability certificate, the flow is borrow-realisable with
just one borrowing, namely of the compound Glyald.  Thus Glyald can be
seen as a network catalyst for this pathway.%

\begin{figure}[tbp]
  \centering
      \resizebox{\textwidth}{!}{%
	  \renewcommand\modInputPrefix{data/ppp_borrow_flow}%
  	% !TEX root = ../../../paper_lncs.tex
\begin{tikzpicture}[abstractNetwork, node distance=2em and 2em]
% id = 0, graphName = Ru5P/X5P
\node[modStyleDGHyperVertex, text depth=0, inputMol] (v-0-0) at (0, 0) {{$\mathrm{Ru5P}$}};
% id = 3{ 'Ru5P/X5P' }, 'Aldose-Ketose <-', { 'R5P' }
\node[modStyleDGHyperEdge, right=of v-0-0] (v-3-0) {$\mathrm{2}$};
% id = 2, graphName = R5P
\node[modStyleDGHyperVertex, right=of v-3-0] (v-2-0) {{$\mathrm{R5P}$}};
% id = 14{ 'R5P' 'Ru5P/X5P' }, 'Transketolase', { 'G3P' 'S7P' }
\node[modStyleDGHyperEdge, below=of v-2-0] (v-14-0) {$\mathrm{2}$};
\node[modStyleDGHyperVertexHidden, yshift = 1.8em, below= of v-14-0] (v-14-0-below) {};
% id = 11, graphName = S7P
\node[modStyleDGHyperVertex, left=of v-14-0] (v-11-0) {{$\mathrm{S7P}$}};
% id = 42{ 'G3P' 'S7P' }, 'Transaldolase', { 'E4P' 'F6P' }
\node[modStyleDGHyperEdge, left=of v-11-0)] (v-42-0) {$\mathrm{2}$};
% id = 35, graphName = E4P
\node[modStyleDGHyperVertex, left=of v-42-0] (v-35-0) {{$\mathrm{E4P}$}};
% id = 86{ 'E4P' 'Ru5P/X5P' }, 'Transketolase', { 'G3P' 'F6P' }
\node[modStyleDGHyperEdge, left=of v-35-0] (v-86-0) {$\mathrm{2}$};
% id = 41, graphName = F6P
\node[modStyleDGHyperVertex, outputMol, above=of v-86-0] (v-41-0) {{$\mathrm{F6P}$}};
% id = 13, graphName = G3P
\node[modStyleDGHyperVertex, left=of v-86-0] (v-13-0) {{$\mathrm{G3P}$}};
% id = 108{ '2Hex' 'G3P' }, 'Transaldolase', { 'Glyald' 'F6P' }
\node[modStyleDGHyperEdge, above=of v-13-0)] (v-108-0) {$\mathrm{1}$};
% id = 22{ 'G3P' }, 'Aldose-Ketose ->', { 'DHAP' }
\node[modStyleDGHyperEdge, left=of v-13-0] (v-22-0) {$\mathrm{1}$};
% id = 21, graphName = DHAP
\node[modStyleDGHyperVertex, left=of v-22-0] (v-21-0) {{$\mathrm{DHAP}$}};
% id = 135{ 'Glyald' 'DHAP' }, 'Aldolase', { 'F1P' }
\node[modStyleDGHyperEdge, above=of v-21-0] (v-135-0) {$\mathrm{1}$};
% id = 10, graphName = Glyald
\node[modStyleDGHyperVertex, left=of v-108-0, text depth=0] (v-10-0) {{$\mathrm{Glyald}$}};

% id = 253{ 'F1P' 'H_2O' }, 'Phophohydrolase', { '2Hex' 'Pi' }
\node[modStyleDGHyperEdge, left=of v-21-0, yshift=1.90em] (v-253-0) {$\mathrm{1}$};
% id = 134, graphName = F1P
\node[modStyleDGHyperVertex, below=of v-253-0] (v-134-0) {{$\mathrm{F1P}$}};
% id = 34, graphName = 2Hex
\node[modStyleDGHyperVertex, above= of v-253-0] (v-34-0) {{$\mathrm{2Hex}$}};
\node[modStyleDGHyperVertexHidden, overlay, above= of v-0-0] (v-0-0-IOFlow) {};
\node[modStyleDGHyperVertexHidden, overlay, above=of v-41-0] (v-41-0-IOFlow) {};
% id = 3{ 'Ru5P/X5P' }, 'Aldose-Ketose <-', { 'R5P' }
\path[modStyleDGHyperConnector] (v-0-0) to (v-3-0);
\path[modStyleDGHyperConnector] (v-3-0) to (v-2-0);
% id = 14{ 'R5P' 'Ru5P/X5P' }, 'Transketolase', { 'G3P' 'S7P' }
\path[modStyleDGHyperConnector] (v-0-0) to (v-14-0);
\path[modStyleDGHyperConnector] (v-2-0) to (v-14-0);
\path[modStyleDGHyperConnector] (v-14-0) to (v-11-0);
\path[modStyleDGHyperConnector, overlay] (v-14-0) to[bend left=30, looseness=0.6] (v-13-0);
% id = 22{ 'G3P' }, 'Aldose-Ketose ->', { 'DHAP' }
\path[modStyleDGHyperConnector] (v-13-0) to (v-22-0);
\path[modStyleDGHyperConnector] (v-22-0) to (v-21-0);
% id = 42{ 'G3P' 'S7P' }, 'Transaldolase', { 'E4P' 'F6P' }
\path[modStyleDGHyperConnector] (v-11-0) to (v-42-0);
\path[modStyleDGHyperConnector] (v-13-0) [bend right=20] to (v-42-0);
\path[modStyleDGHyperConnector] (v-42-0) to (v-35-0);
\path[modStyleDGHyperConnector] (v-42-0) to (v-41-0);
% id = 86{ 'E4P' 'Ru5P/X5P' }, 'Transketolase', { 'G3P' 'F6P' }
\path[modStyleDGHyperConnector] (v-0-0) to (v-86-0);
\path[modStyleDGHyperConnector] (v-35-0) to (v-86-0);
\path[modStyleDGHyperConnector] (v-86-0) to (v-13-0);
\path[modStyleDGHyperConnector] (v-86-0) to (v-41-0);
% id = 108{ '2Hex' 'G3P' }, 'Transaldolase', { 'Glyald' 'F6P' }
\path[modStyleDGHyperConnector] (v-13-0) to (v-108-0);
\path[modStyleDGHyperConnector] (v-34-0) to[out=0, in=150, in looseness=0.5] (v-108-0);
\path[modStyleDGHyperConnector] (v-108-0) to (v-10-0);
\path[modStyleDGHyperConnector] (v-108-0) to (v-41-0);
% id = 135{ 'Glyald' 'DHAP' }, 'Aldolase', { 'F1P' }
\path[modStyleDGHyperConnector] (v-10-0) to (v-135-0);
\path[modStyleDGHyperConnector] (v-21-0) to (v-135-0);
\path[modStyleDGHyperConnector] (v-135-0) to[out=-160, in=30] (v-134-0);
% id = 253{ 'F1P' 'H_2O' }, 'Phophohydrolase', { '2Hex' 'Pi' }
\path[modStyleDGHyperConnector] (v-134-0) to (v-253-0);
\path[modStyleDGHyperConnector] (v-253-0) to (v-34-0);
% inFlow/outFlow, id = 0-0, graphName = Ru5P/X5P, inFlow = 6, outFlow = 0
\path[modStyleDGHyperConnector] (v-0-0-IOFlow) to node[auto, swap] {$\mathrm{6}$} (v-0-0);
% inFlow/outFlow, id = 41-0, graphName = F6P, inFlow = 0, outFlow = 5
\path[modStyleDGHyperConnector] (v-41-0) to node[auto, swap] {$\mathrm{5}$} (v-41-0-IOFlow);
\end{tikzpicture}%
%
  	}%
      \caption{Example of a flow for the pentose phosphate pathway
        that is not scaled-realisable. The flow is
        borrow-realisable. The input compound is marked with green and the
        output compound is marked with blue. The SMILES strings for all molecule identifiers are listed in Appendix, Table~\ref{tab:smiles}.}
      \label{fig:borrow_realisable_flow_PPP}
\end{figure}

\begin{figure}[tbp]
  \centering
  \resizebox{\textwidth}{!}{%
	  \renewcommand\modInputPrefix{data/ppp_borrow_cert}%
  	% !TEX root = ../../../paper_lncs.tex
{%
% nodeToMake, topNode, direction
\newcommand\makeMid[3]{%
	\node[modStyleDGHyperVertexHidden, #3=1.5em of #2] (#1-top) {};
	\node[modStyleDGHyperVertexHidden, below=of #1-top] (#1-bottom) {};
	\path (#1-top)--(#1-bottom) node[modStyleDGHyperEdge, midway] (#1) {};
}
\newcommand\makeMidAbove[3]{%
	\node[modStyleDGHyperVertexHidden, #3=1.5em of #2] (#1-bottom) {};
	\node[modStyleDGHyperVertexHidden, above=of #1-bottom] (#1-top) {};
	\path (#1-top)--(#1-bottom) node[modStyleDGHyperEdge, midway] (#1) {};
}
\begin{tikzpicture}[abstractNetwork, node distance=1em and 2em]
% id = 0, graphName = Ru5P
\node[modStyleDGHyperVertex, text depth=0, inputMol] (v-0-10) {{$\mathrm{Ru5P}$}};
% id = 3{ 'Ru5P' }, 'Aldose-Ketose <-', { 'R5P' }
\node[modStyleDGHyperEdge, right= of v-0-10] (v-3-15) {};
% id = 2, graphName = R5P
\node[modStyleDGHyperVertex, right= of v-3-15] (v-2-17) {{$\mathrm{R5P}$}};
% id = 0, graphName = Ru5P
\node[modStyleDGHyperVertex, text depth=0, inputMol, below=of v-2-17] (v-0-12) {{$\mathrm{Ru5P}$}};
% id = 14{ 'R5P' 'Ru5P' }, 'Transketolase', { 'G3P' 'S7P' }
\makeMid{v-14-19}{v-2-17}{right}
% id = 11, graphName = S7P
\node[modStyleDGHyperVertex, right=1.5em of v-14-19-top] (v-11-21) {{$\mathrm{S7P}$}};
% id = 13, graphName = G3P
\node[modStyleDGHyperVertex, below=of v-11-21] (v-13-22) {{$\mathrm{G3P}$}};
% id = 42{ 'G3P' 'S7P' }, 'Transaldolase', { 'E4P' 'F6P' }
\makeMid{v-42-25}{v-11-21}{right}
% id = 41, graphName = F6P
\node[modStyleDGHyperVertex, outputMol, right=1.5em of v-42-25-top] (v-41-28) {{$\mathrm{F6P}$}};
% id = 35, graphName = E4P
\node[modStyleDGHyperVertex, below= of v-41-28] (v-35-27) {{$\mathrm{E4P}$}};
% id = 0, graphName = Ru5P
\node[modStyleDGHyperVertex, text depth=0, inputMol, below= of v-35-27] (v-0-14) {{$\mathrm{Ru5P}$}};
% id = 86{ 'E4P' 'Ru5P' }, 'Transketolase', { 'G3P' 'F6P' }
\makeMid{v-86-31}{v-35-27}{right}
% id = 41, graphName = F6P
\node[modStyleDGHyperVertex, outputMol, right=1.5em of v-86-31-top] (v-41-35) {{$\mathrm{F6P}$}};
% id = 13, graphName = G3P
\node[modStyleDGHyperVertex, below=of v-41-35] (v-13-36) {{$\mathrm{G3P}$}};

% id = 0, graphName = Ru5P
\node[modStyleDGHyperVertex, text depth=0, inputMol, below=14em of v-0-10] (v-0-11) {{$\mathrm{Ru5P}$}};
% id = 3{ 'Ru5P' }, 'Aldose-Ketose <-', { 'R5P' }
\node[modStyleDGHyperEdge, right=of v-0-11] (v-3-16) {};
% id = 2, graphName = R5P
\node[modStyleDGHyperVertex, right=of v-3-16] (v-2-18) {{$\mathrm{R5P}$}};
% id = 0, graphName = Ru5P
\node[modStyleDGHyperVertex, text depth=0, inputMol, above=of v-2-18] (v-0-13) {{$\mathrm{Ru5P}$}};
% id = 14{ 'R5P' 'Ru5P' }, 'Transketolase', { 'G3P' 'S7P' }
\makeMidAbove{v-14-20}{v-2-18}{right}
% id = 11, graphName = S7P
\node[modStyleDGHyperVertex, right=1.5em of v-14-20-bottom] (v-11-23) {{$\mathrm{S7P}$}};
% id = 13, graphName = G3P
\node[modStyleDGHyperVertex, above=of v-11-23] (v-13-24) {{$\mathrm{G3P}$}};
% id = 42{ 'G3P' 'S7P' }, 'Transaldolase', { 'E4P' 'F6P' }
\makeMidAbove{v-42-26}{v-11-23}{right}
% id = 41, graphName = F6P
\node[modStyleDGHyperVertex, outputMol, right=1.5em of v-42-26-bottom] (v-41-30) {{$\mathrm{F6P}$}};
% id = 35, graphName = E4P
\node[modStyleDGHyperVertex, above=of v-41-30] (v-35-29) {{$\mathrm{E4P}$}};
% id = 0, graphName = Ru5P
\node[modStyleDGHyperVertex, text depth=0, inputMol, above=of v-35-29] (v-0-8) {{$\mathrm{Ru5P}$}};
% id = 86{ 'E4P' 'Ru5P' }, 'Transketolase', { 'G3P' 'F6P' }
\makeMidAbove{v-86-33}{v-35-29}{right}
% id = 41, graphName = F6P
\node[modStyleDGHyperVertex, outputMol, right=1.5em of v-86-33-bottom] (v-41-37) {{$\mathrm{F6P}$}};
% id = 13, graphName = G3P
\node[modStyleDGHyperVertex,  above=of v-41-37] (v-13-38) {{$\mathrm{G3P}$}};

% id = 108{ '2Hex' 'G3P' }, 'Transaldolase', { 'Glyald' 'F6P' }
\node[modStyleDGHyperEdge, at=($(v-13-38)!0.5!(v-13-36)$)] (v-108-49) {};
% id = 10, graphName = Glyald
\node[modStyleDGHyperVertex, borrowMol, at=(v-108-49), shift={(3.5em, -1.2em)}, text depth=0] (v-10-50) {{$\mathrm{Glyald}$}};
% id = 41, graphName = F6P
\node[modStyleDGHyperVertex, outputMol, at=(v-108-49), shift={(3.5em, +1.2em)},] (v-41-51) {{$\mathrm{F6P}$}};

% id = 34, graphName = 2Hex
\node[modStyleDGHyperVertex, left=of v-108-49] (v-34-47) {{$\mathrm{2Hex}$}};
% id = 253{ 'F1P' 'H_2O' }, 'Phophohydrolase', { '2Hex' 'Pi' }
\node[modStyleDGHyperEdge, left=of v-34-47] (v-253-45) {};
% id = 134, graphName = F1P
\node[modStyleDGHyperVertex, left=of v-253-45] (v-134-44) {{$\mathrm{F1P}$}};
% id = 135{ 'Glyald' 'DHAP' }, 'Aldolase', { 'F1P' }
\node[modStyleDGHyperEdge, left=of v-134-44] (v-135-43) {};
% id = 21, graphName = DHAP
\node[modStyleDGHyperVertex, above=of v-135-43] (v-21-42) {{$\mathrm{DHAP}$}};
% id = 22{ 'G3P' }, 'Aldose-Ketose ->', { 'DHAP' }
\node[modStyleDGHyperEdge, right=of v-21-42] (v-22-40) {};

% id = 10, graphName = Glyald
\node[modStyleDGHyperVertex, text depth=0, borrowMol, left=of v-135-43] (v-10-0) {{$\mathrm{Glyald}$}};

% id = 3{ 'Ru5P' }, 'Aldose-Ketose <-', { 'R5P' }
\path[modStyleDGHyperConnector] (v-0-10) to (v-3-15);
\path[modStyleDGHyperConnector] (v-3-15) to (v-2-17);
\path[modStyleDGHyperConnector] (v-0-11) to (v-3-16);
\path[modStyleDGHyperConnector] (v-3-16) to (v-2-18);
% id = 14{ 'R5P' 'Ru5P' }, 'Transketolase', { 'G3P' 'S7P' }
\path[modStyleDGHyperConnector] (v-0-12) to (v-14-19);
\path[modStyleDGHyperConnector] (v-2-17) to (v-14-19);
\path[modStyleDGHyperConnector] (v-14-19) to (v-11-21);
\path[modStyleDGHyperConnector] (v-14-19) to (v-13-22);
\path[modStyleDGHyperConnector] (v-0-13) to (v-14-20);
\path[modStyleDGHyperConnector] (v-2-18) to (v-14-20);
\path[modStyleDGHyperConnector] (v-14-20) to (v-11-23);
\path[modStyleDGHyperConnector] (v-14-20) to (v-13-24);
% id = 22{ 'G3P' }, 'Aldose-Ketose ->', { 'DHAP' }
\path[modStyleDGHyperConnector] (v-13-36) to[in=0, out=-170] (v-22-40);
\path[modStyleDGHyperConnector] (v-22-40) to (v-21-42);
% id = 42{ 'G3P' 'S7P' }, 'Transaldolase', { 'E4P' 'F6P' }
\path[modStyleDGHyperConnector] (v-11-21) to (v-42-25);
\path[modStyleDGHyperConnector] (v-13-22) to (v-42-25);
\path[modStyleDGHyperConnector] (v-42-25) to (v-35-27);
\path[modStyleDGHyperConnector] (v-42-25) to (v-41-28);
\path[modStyleDGHyperConnector] (v-11-23) to (v-42-26);
\path[modStyleDGHyperConnector] (v-13-24) to (v-42-26);
\path[modStyleDGHyperConnector] (v-42-26) to (v-35-29);
\path[modStyleDGHyperConnector] (v-42-26) to (v-41-30);
% id = 86{ 'E4P' 'Ru5P' }, 'Transketolase', { 'G3P' 'F6P' }
\path[modStyleDGHyperConnector] (v-0-14) to (v-86-31);
\path[modStyleDGHyperConnector] (v-35-27) to (v-86-31);
\path[modStyleDGHyperConnector] (v-86-31) to (v-13-36);
\path[modStyleDGHyperConnector] (v-86-31) to (v-41-35);
\path[modStyleDGHyperConnector] (v-0-8) to (v-86-33);
\path[modStyleDGHyperConnector] (v-35-29) to (v-86-33);
\path[modStyleDGHyperConnector] (v-86-33) to (v-13-38);
\path[modStyleDGHyperConnector] (v-86-33) to (v-41-37);
% id = 108{ '2Hex' 'G3P' }, 'Transaldolase', { 'Glyald' 'F6P' }
\path[modStyleDGHyperConnector] (v-13-38) to (v-108-49);
\path[modStyleDGHyperConnector] (v-34-47) to (v-108-49);
\path[modStyleDGHyperConnector] (v-108-49) to (v-10-50);
\path[modStyleDGHyperConnector] (v-108-49) to (v-41-51);
% id = 135{ 'Glyald' 'DHAP' }, 'Aldolase', { 'F1P' }
\path[modStyleDGHyperConnector] (v-10-0) to (v-135-43);
\path[modStyleDGHyperConnector] (v-21-42) to (v-135-43);
\path[modStyleDGHyperConnector] (v-135-43) to (v-134-44);
% id = 253{ 'F1P' 'H_2O' }, 'Phophohydrolase', { '2Hex' 'Pi' }
%\path[modStyleDGHyperConnector] (v-1-9) to (v-253-45);
\path[modStyleDGHyperConnector] (v-134-44) to (v-253-45);
%\path[modStyleDGHyperConnector] (v-253-45) to (v-4-46);
\path[modStyleDGHyperConnector] (v-253-45) to (v-34-47);
\end{tikzpicture}%
}%
%
  }%	
  \caption{A realisability certificate for the flow in
    Fig.~\ref{fig:borrow_realisable_flow_PPP} where the molecule Glyald is
    borrowed in order to make it borrow-realisable. The input compounds are
    marked with green, the output compounds are marked with blue and the
    borrowed compound is marked with purple. The SMILES strings for all molecule identifiers are listed in Appendix, Table~\ref{tab:smiles}.}
  \label{fig:borrow_realisable_DAG_PPP}
\end{figure}

\section{Representations of Pathways}
\label{sec:pathwayrep}
We have described two ways of modelling pathways:
flows and realisability certificates. The realisability
certificate defines a causal order in the pathway and explicitly
expresses which individual molecule is used when and for which
reaction. A realisability certificate uniquely determines a
corresponding flow.
Flows, on the other
hand, do not specify the order of the reactions or which one of
multiple copies of a molecules is used in which reaction.  A flow therefore may correspond to multiple different realisability
certificates, each representing a different mechanism.

\begin{figure}[tbp]
  \centering
  \resizebox{0.4\textwidth}{!}{%
  \begin{tikzpicture}[abstractNetwork, node distance = 2em and 2em]
  \tikzstyle{every node}=[font=\small]
  \node[circle, minimum size = 7.7cm] (c) at (0,0) {};
  \node at (c.90) (E4P){E4P};
  \node at (c.50) (S7P) {S7P};
  \node at (c.10) (R5P) {R5P};
  \node at (c.330) (Ru5P) {Ru5P};
  \node at (c.290) (X5P) {X5P};
  \node at (c.250) (G3P) {G3P};
  \node at (c.210) (DHAP) {DHAP};
  \node at (c.170) (FBP) {FBP};
  \node at (c.130) (F6P) {F6P};
  \node at (c.70) (mid_1) {};
  \node[xshift=-0.05em] at (c.30) (mid_2) {};
  \node at (c.270) (mid_3) {};
  \node[xshift=0.05em] at (c.190) (mid_5) {};
  \node at (c.110) (mid_6) {};
  \draw[arrow] (E4P) [bend left=14]  to (S7P);
    \draw[arrow] (S7P) [bend left = 14] to (R5P);
    \draw[arrow] (R5P) [bend left = 14] to (Ru5P);
    \draw[arrow] (Ru5P) [bend left = 14] to (X5P);
    \draw[arrow] (X5P) [bend left = 14] to (G3P);
    \draw[arrow] (G3P) [bend left = 14] to (DHAP);
    \draw[arrow] (DHAP) [bend left = 14] to (FBP);
    \draw[arrow] (FBP) [bend left = 14] to (F6P);
    \draw[arrow] (F6P) [bend left = 14] to (E4P);
    \node[above right = of E4P, xshift = -1em] (F6P1) {F6P};
    \node[yshift=2.5em] at(mid_3)(mid_small_circ) {};
    \node at (mid_small_circ) [circle through ={(mid_3)}] (small_circ) {};
    \node at (small_circ.0) (E4P1) {E4P};
    \node at (small_circ.180) (F6P2) {F6P};
    \node[above left= of E4P, xshift = 1em] (AcP1) {AcP};
    \node at (small_circ.90) (mid_4) {};
    \node[above =of mid_4] (AcP2) {AcP};
    \draw[semithick] (E4P1) [bend left=35] to (mid_3.center);
    \draw[arrow] (mid_3.center) [bend left=35] to (F6P2);
    \draw[semithick] (F6P2) [bend left = 35] to (mid_4.center);
    \draw[arrow] (mid_4.center) [bend left = 35] to (E4P1);
    \draw[arrow] (mid_4.center) [bend right] to (AcP2);
    \draw[arrow]([yshift=0.02em]mid_6.center) [bend right=35] to (AcP1);
    \draw ($(mid_2)!3.5em!280:(R5P)$) node (mid_small_circ_3) {};
    \node at (mid_small_circ_3) [circle through ={(mid_2.center)}] (small_circ_3) {};
    \node at (small_circ_3.120) (G3P1) {G3P};
    \node at (small_circ_3.300) (X5P1) {X5P};
    \draw[semithick] (F6P1)  [bend right=35] to ([yshift=0.02em]mid_1.center);
    \draw[arrow] ([yshift=0.025em]mid_1.center) [bend left=35] to (G3P1);
    \draw[semithick] (G3P1) [bend left=30] to (mid_2.center);
    \draw[arrow] (mid_2.center) [bend left=30] to (X5P1);
    \draw ($(mid_5)!3.5em!280:(FBP)$) node (mid_small_circ_4) {};
    \node at (mid_small_circ_4) [circle through ={(mid_5)}] (small_circ_4) {};
    \node at (small_circ_4.300) (G3P2) {G3P};
    \draw[arrow] (X5P1) -- (G3P2) node[midway] (mid_7){};
    \draw[semithick] (G3P2) [bend left=37] to ([xshift=-0.017em]mid_5.center);
    \draw ($(mid_7)!3.5em!90:(X5P1)$) node (mid_small_circ_2) {};
    \node at (mid_small_circ_2) [circle through ={(mid_7)}] (small_circ_2) {};
    \node at (small_circ_2.200) (F6P3) {F6P};
    \node at (small_circ_2.20) (E4P2) {E4P};
    \node at (small_circ_2.110) (mid_8) {};
    \node[above = of mid_8] (AcP3){AcP};
    \draw[arrow](mid_7.center) [bend left=35] to (F6P3);
    \draw[semithick] (E4P2) [bend left=35] to (mid_7.center);
    \draw[semithick] (F6P3) [bend left=35] to (mid_8.center);
    \draw[arrow] (mid_8.center) [bend left=35] to (E4P2);
    \draw[arrow] (mid_8.center) [bend right] to (AcP3);
     \end{tikzpicture}%
     }%
  \caption{Example of a pathway drawing for the cyclic non-oxidative
  glycolysis (NOG) pathway. Recreated from \cite[Fig. 2a]{bogorad:13}.}
\label{fig:pathway_drawing}
\end{figure}

We want to point out that commonly used representations of pathways
in the life science literature fall in between these two
extremes, see Fig.~\ref{fig:pathway_drawing} for an example. In this example, the order of reactions is not fully
resolved---for instance, is F6P produced before E4P or after? Indeed, some unspecified choice of borrowing is needed to set
the pathway in motion. Additionally, the semantics of a molecule
identifier appearing in several places is unclear---for instance,
are the three appearances of G3P interchangeable in the associated
reactions or do they signify different individual instances of the
same type of molecule? In the former case, the figure corresponds to
a much larger number of different realisability certificates than in
the latter case. The answers to these questions have important
consequences for investigations where the identity of individual
atoms matter, such as atom tracing.

Furthermore, when there is a choice between different pathway
suggestions, avoiding borrow-realisable pathways often gives simpler
depictions. However, this introduces a bias among the possible
pathways, which may be unwanted, as borrow-realisable solutions are
usually equally simple in chemical terms.
We note that the need for borrowing in pathways is usually not
discussed in the literature. Additionally, there has been a lack of
computational methods to systematically look for borrow-realisable
pathways, even if they could equally likely form part of what happens
in nature. For instance, the PPP is usually depicted in a form that
give rise to a realisable flow depicted in Fig.~\ref{fig:realisable_flow_PPP_id}, 
with a realisability certificate
shown in Fig.~\ref{fig:realisable_DAG_PPP_id}. It could just as well be
described by the equally simple and chemically realistic borrow-realisable pathway depicted in
Fig.~\ref{fig:borrow_realisable_flow_PPP}.%

\begin{figure}[tbp]
\centering
\resizebox{\textwidth}{!}{%
   \renewcommand\modInputPrefix{data/ppp_realisable_flow_id}%
  	% !TEX root = ../../../paper_lncs.tex
\begin{tikzpicture}[remember picture, scale=\modDGHyperScale]
% id = 1, graphName = Ru5P/X5P
\node[modStyleDGHyperVertex,  text depth=0, inputMol] (v-1-0) {{$\mathrm{Ru5P}$}};
% id = 3{ 'Ru5P/X5P' }, 'Aldose-Ketose <-', { 'R5P' }
\node[modStyleDGHyperEdge, right = of v-1-0] (v-3-0) {$\mathrm{2}$};
% id = 2, graphName = R5P
\node[modStyleDGHyperVertex, right = of v-3-0] (v-2-0) {{$\mathrm{R5P}$}};
% id = 6{ 'R5P' 'Ru5P/X5P' }, 'Transketolase', { 'G3P' 'S7P' }
\node[modStyleDGHyperEdge, below = of v-2-0] (v-6-0) {$\mathrm{2}$};
\node[modStyleDGHyperVertexHidden, yshift=2.1em, below= of v-6-0] (v-6-0-below) {};
% id = 4, graphName = S7P
\node[modStyleDGHyperVertex, below = of v-3-0] (v-4-0) {{$\mathrm{S7P}$}};
% id = 23{ 'G3P' 'S7P' }, 'Transaldolase', { 'E4P' 'F6P' }
\node[modStyleDGHyperEdge, below = of v-1-0] (v-23-0) {$\mathrm{2}$};
% id = 22, graphName = E4P
\node[modStyleDGHyperVertex, left = of v-23-0] (v-22-0) {{$\mathrm{E4P}$}};
% id = 34{ 'E4P' 'Ru5P/X5P' }, 'Transketolase', { 'G3P' 'F6P' }
\node[modStyleDGHyperEdge, left = of v-22-0] (v-34-0) {$\mathrm{2}$};
% id = 5, graphName = G3P
\node[modStyleDGHyperVertex, left = of v-34-0] (v-5-0) {{$\mathrm{G3P}$}};
% id = 8{ 'G3P' }, 'Aldose-Ketose ->', { 'DHAP' }
\node[modStyleDGHyperEdge, left = of v-5-0] (v-8-0) {$\mathrm{1}$};
% id = 7, graphName = DHAP
\node[modStyleDGHyperVertex, left = of v-8-0] (v-7-0) {{$\mathrm{DHAP}$}};
% id = 12{ 'G3P' 'DHAP' }, 'Aldolase', { 'FBP' }
\node[modStyleDGHyperEdge, above = of v-7-0] (v-12-0) {$\mathrm{1}$};
% id = 11, graphName = FBP
\node[modStyleDGHyperVertex, above = of v-8-0] (v-11-0) {{$\mathrm{FBP}$}};
% id = 21{ 'FBP' 'H_2O' }, 'Phophohydrolase', { 'Pi' 'F6P' }
\node[modStyleDGHyperEdge, above = of v-5-0] (v-21-0) {$\mathrm{1}$};
% id = 20, graphName = F6P
\node[modStyleDGHyperVertex, above = of v-34-0, outputMol] (v-20-0) {{$\mathrm{F6P}$}};
% id = 3{ 'Ru5P/X5P' }, 'Aldose-Ketose <-', { 'R5P' }
\path[modStyleDGHyperConnector] (v-1-0) to (v-3-0);
\path[modStyleDGHyperConnector] (v-3-0) to (v-2-0);
% id = 6{ 'R5P' 'Ru5P/X5P' }, 'Transketolase', { 'G3P' 'S7P' }
\path[modStyleDGHyperConnector] (v-1-0) to (v-6-0);
\path[modStyleDGHyperConnector] (v-2-0) to (v-6-0);
\path[modStyleDGHyperConnector] (v-6-0) to (v-4-0);
\path[modStyleDGHyperConnector, overlay] (v-6-0) to[bend left=30, looseness=0.6] (v-5-0);
% id = 8{ 'G3P' }, 'Aldose-Ketose ->', { 'DHAP' }
\path[modStyleDGHyperConnector] (v-5-0) to (v-8-0);
\path[modStyleDGHyperConnector] (v-8-0) to (v-7-0);
% id = 12{ 'G3P' 'DHAP' }, 'Aldolase', { 'FBP' }
\path[modStyleDGHyperConnector] (v-5-0) to (v-12-0);
\path[modStyleDGHyperConnector] (v-7-0) to (v-12-0);
\path[modStyleDGHyperConnector] (v-12-0) to (v-11-0);
% id = 21{ 'FBP' 'H_2O' }, 'Phophohydrolase', { 'Pi' 'F6P' }
\path[modStyleDGHyperConnector] (v-11-0) to (v-21-0);
\path[modStyleDGHyperConnector] (v-21-0) to (v-20-0);
% id = 23{ 'G3P' 'S7P' }, 'Transaldolase', { 'E4P' 'F6P' }
\path[modStyleDGHyperConnector] (v-4-0) to (v-23-0);
\path[modStyleDGHyperConnector] (v-5-0) to [bend right=20] (v-23-0);
\path[modStyleDGHyperConnector] (v-23-0) to (v-20-0);
\path[modStyleDGHyperConnector] (v-23-0) to (v-22-0);
% id = 34{ 'E4P' 'Ru5P/X5P' }, 'Transketolase', { 'G3P' 'F6P' }
\path[modStyleDGHyperConnector] (v-1-0) to (v-34-0);
\path[modStyleDGHyperConnector] (v-22-0) to (v-34-0);
\path[modStyleDGHyperConnector] (v-34-0) to (v-5-0);
\path[modStyleDGHyperConnector] (v-34-0) to (v-20-0);
% inFlow/outFlow, id = 1-0, graphName = Ru5P/X5P, inFlow = 6, outFlow = 0
\node[modStyleDGHyperVertexHidden, overlay, above = of v-1-0] (v-1-0-IOFlow) {};
\path[modStyleDGHyperConnector] (v-1-0-IOFlow) to node[auto, swap] {$\mathrm{6}$} (v-1-0);
% inFlow/outFlow, id = 20-0, graphName = F6P, inFlow = 0, outFlow = 5
\node[modStyleDGHyperVertexHidden, overlay, above = of v-20-0] (v-20-0-IOFlow) {};
\path[modStyleDGHyperConnector] (v-20-0) to node[auto, swap] {$\mathrm{5}$} (v-20-0-IOFlow);
\end{tikzpicture}%
%
  	}%
\caption{A flow for the pentose phosphate pathway which is realisable. The input compound is marked with green and the output compound is marked with blue. The SMILES strings for all molecule identifiers are listed in Appendix, Table~\ref{tab:smiles}.}
\label{fig:realisable_flow_PPP_id}
\end{figure}

\begin{figure}[tbp]
\centering
\resizebox{\textwidth}{!}{%
  \renewcommand\modInputPrefix{data/ppp_realisable_cert_id}%
  	% !TEX root = ../../../paper_lncs.tex
{%
% nodeToMake, topNode, direction
\newcommand\makeMid[3]{%
	\node[modStyleDGHyperVertexHidden, #3=1.5em of #2] (#1-top) {};
	\node[modStyleDGHyperVertexHidden, below=of #1-top] (#1-bottom) {};
	\path (#1-top)--(#1-bottom) node[modStyleDGHyperEdge, midway] (#1) {};
}
\newcommand\makeMidAbove[3]{%
	\node[modStyleDGHyperVertexHidden, #3=1.5em of #2] (#1-bottom) {};
	\node[modStyleDGHyperVertexHidden, above=of #1-bottom] (#1-top) {};
	\path (#1-top)--(#1-bottom) node[modStyleDGHyperEdge, midway] (#1) {};
}
\begin{tikzpicture}[abstractNetwork, node distance=1em and 2em]
% id = 1, graphName = Ru5P/X5P
\node[modStyleDGHyperVertex,inputMol, text depth=0] (v-1-10) {{$\mathrm{Ru5P}$}};
% id = 3{ 'Ru5P/X5P' }, 'Aldose-Ketose <-', { 'R5P' }
\node[modStyleDGHyperEdge, right = of v-1-10] (v-3-15) {};
% id = 2, graphName = R5P
\node[modStyleDGHyperVertex, right = of v-3-15] (v-2-17) {{$\mathrm{R5P}$}};
% id = 6{ 'R5P' 'Ru5P/X5P' }, 'Transketolase', { 'G3P' 'S7P' }
%\node[modStyleDGHyperEdge, at=(v-coord-6-19)] (v-6-19) {};
\makeMid{v-6-19}{v-2-17}{right}
% id = 1, graphName = Ru5P/X5P
\node[modStyleDGHyperVertex, text depth=0, inputMol, below= of v-2-17] (v-1-12) {{$\mathrm{Ru5P}$}};
% id = 4, graphName = S7P
\node[modStyleDGHyperVertex, right = 1.5em of v-6-19-top] (v-4-22) {{$\mathrm{S7P}$}};
% id = 5, graphName = G3P
\node[modStyleDGHyperVertex, below = of v-4-22)] (v-5-23) {{$\mathrm{G3P}$}};
% id = 23{ 'G3P' 'S7P' }, 'Transaldolase', { 'E4P' 'F6P' }
%\node[modStyleDGHyperEdge, above = of v-5-23] (v-23-24) {};
\makeMid{v-23-24}{v-5-23}{right}
% id = 20, graphName = F6P
\node[modStyleDGHyperVertex, outputMol, above = of v-23-24] (v-20-26) {{$\mathrm{F6P}$}};
% id = 22, graphName = E4P
\node[modStyleDGHyperVertex,right = 1.5em of v-23-24-bottom] (v-22-27) {{$\mathrm{E4P}$}};
% id = 34{ 'E4P' 'Ru5P/X5P' }, 'Transketolase', { 'G3P' 'F6P' }
%\node[modStyleDGHyperEdge, at=(v-coord-34-30)] (v-34-30) {};
\makeMidAbove{v-34-30}{v-22-27}{right}
% id = 1, graphName = Ru5P/X5P
\node[modStyleDGHyperVertex, above = of v-22-27, inputMol, text depth=0] (v-1-13) {{$\mathrm{Ru5P}$}};
% id = 5, graphName = G3P
\node[modStyleDGHyperVertex,  right = 1.5em of v-34-30-top] (v-5-32) {{$\mathrm{G3P}$}};
% id = 20, graphName = F6P
\node[modStyleDGHyperVertex, below = of v-5-32, outputMol] (v-20-31) {{$\mathrm{F6P}$}};
% id = 23{ 'G3P' 'S7P' }, 'Transaldolase', { 'E4P' 'F6P' }
\node[modStyleDGHyperEdge, above = 1.3em of v-5-32] (v-23-34) {};
% id = 20, graphName = F6P
\node[modStyleDGHyperVertex, above = of v-23-34, outputMol] (v-20-35) {{$\mathrm{F6P}$}};
% id = 22, graphName = E4P
\node[modStyleDGHyperVertex, right = of v-23-34] (v-22-36) {{$\mathrm{E4P}$}};
% id = 34{ 'E4P' 'Ru5P/X5P' }, 'Transketolase', { 'G3P' 'F6P' }
%\node[modStyleDGHyperEdge, at=(v-coord-34-38)] (v-34-38) {};
\makeMid{v-34-38}{v-22-36}{right}
% id = 1, graphName = Ru5P/X5P
\node[modStyleDGHyperVertex,  below = of v-22-36, inputMol, text depth = 0] (v-1-7) {{$\mathrm{Ru5P}$}};
% id = 20, graphName = F6P
\node[modStyleDGHyperVertex, right = 1.5em of v-34-38-top, outputMol] (v-20-39) {{$\mathrm{F6P}$}};
% id = 5, graphName = G3P
\node[modStyleDGHyperVertex, below = of v-20-39] (v-5-40) {{$\mathrm{G3P}$}};

% id = 1, graphName = Ru5P/X5P
\node[modStyleDGHyperVertex, below = of v-1-12, inputMol, text depth = 0] (v-1-11) {{$\mathrm{Ru5P}$}};
% id = 2, graphName = R5P
\node[modStyleDGHyperVertex, below = of v-1-11] (v-2-16) {{$\mathrm{R5P}$}};
% id = 3{ 'Ru5P/X5P' }, 'Aldose-Ketose <-', { 'R5P' }
\node[modStyleDGHyperEdge, left = of v-2-16] (v-3-14) {};
% id = 1, graphName = Ru5P/X5P
\node[modStyleDGHyperVertex, text depth=0,inputMol, left = of v-3-14] (v-1-9) {{$\mathrm{Ru5P}$}};

% id = 6{ 'R5P' 'Ru5P/X5P' }, 'Transketolase', { 'G3P' 'S7P' }
\makeMid{v-6-18}{v-1-11}{right}
%\node[modStyleDGHyperEdge, right = of v-2-16] (v-6-18) {};
% id = 4, graphName = S7P
\node[modStyleDGHyperVertex, below = of v-5-23] (v-4-20) {{$\mathrm{S7P}$}};
% id = 5, graphName = G3P
\node[modStyleDGHyperVertex, below = of v-4-20] (v-5-21) {{$\mathrm{G3P}$}};
% id = 8{ 'G3P' }, 'Aldose-Ketose ->', { 'DHAP' }
\node[modStyleDGHyperEdge, below = 1.4em of v-23-24-bottom] (v-8-25) {};
% id = 7, graphName = DHAP
\node[modStyleDGHyperVertex, below = 1.05em of v-22-27] (v-7-28) {{$\mathrm{DHAP}$}};
% id = 12{ 'G3P' 'DHAP' }, 'Aldolase', { 'FBP' }
\node[modStyleDGHyperEdge, below = 1.4em of v-34-30-bottom] (v-12-42) {};
% id = 11, graphName = FBP
\node[modStyleDGHyperVertex, right = of v-12-42] (v-11-43) {{$\mathrm{FBP}$}};
% id = 21{ 'FBP' 'H_2O' }, 'Phophohydrolase', { 'Pi' 'F6P' }
\node[modStyleDGHyperEdge, right = of v-11-43] (v-21-44) {};
% id = 20, graphName = F6P
\node[modStyleDGHyperVertex, right = of v-21-44, outputMol] (v-20-46) {{$\mathrm{F6P}$}};

% id = 3{ 'Ru5P/X5P' }, 'Aldose-Ketose <-', { 'R5P' }
\path[modStyleDGHyperConnector] (v-1-9) to (v-3-14);
\path[modStyleDGHyperConnector] (v-3-14) to (v-2-16);
\path[modStyleDGHyperConnector] (v-1-10) to (v-3-15);
\path[modStyleDGHyperConnector] (v-3-15) to (v-2-17);
% id = 6{ 'R5P' 'Ru5P/X5P' }, 'Transketolase', { 'G3P' 'S7P' }
\path[modStyleDGHyperConnector] (v-1-11) to (v-6-18);
\path[modStyleDGHyperConnector] (v-2-16) to (v-6-18);
\path[modStyleDGHyperConnector] (v-6-18) to (v-4-20);
\path[modStyleDGHyperConnector] (v-6-18) to (v-5-21);
\path[modStyleDGHyperConnector] (v-1-12) to (v-6-19);
\path[modStyleDGHyperConnector] (v-2-17) to (v-6-19);
\path[modStyleDGHyperConnector] (v-6-19) to (v-4-22);
\path[modStyleDGHyperConnector] (v-6-19) to (v-5-23);
% id = 8{ 'G3P' }, 'Aldose-Ketose ->', { 'DHAP' }
\path[modStyleDGHyperConnector] (v-5-21) to (v-8-25);
\path[modStyleDGHyperConnector] (v-8-25) to (v-7-28);
% id = 12{ 'G3P' 'DHAP' }, 'Aldolase', { 'FBP' }
\path[modStyleDGHyperConnector] (v-5-40) to (v-12-42);
\path[modStyleDGHyperConnector] (v-7-28) to (v-12-42);
\path[modStyleDGHyperConnector] (v-12-42) to (v-11-43);
% id = 21{ 'FBP' 'H_2O' }, 'Phophohydrolase', { 'Pi' 'F6P' }
\path[modStyleDGHyperConnector] (v-11-43) to (v-21-44);
\path[modStyleDGHyperConnector] (v-21-44) to (v-20-46);
% id = 23{ 'G3P' 'S7P' }, 'Transaldolase', { 'E4P' 'F6P' }
\path[modStyleDGHyperConnector] (v-4-20) to (v-23-24);
\path[modStyleDGHyperConnector] (v-5-23) to (v-23-24);
\path[modStyleDGHyperConnector] (v-23-24) to (v-20-26);
\path[modStyleDGHyperConnector] (v-23-24) to (v-22-27);
\path[modStyleDGHyperConnector] (v-4-22) to (v-23-34);
\path[modStyleDGHyperConnector] (v-5-32) to (v-23-34);
\path[modStyleDGHyperConnector] (v-23-34) to (v-20-35);
\path[modStyleDGHyperConnector] (v-23-34) to (v-22-36);
% id = 34{ 'E4P' 'Ru5P/X5P' }, 'Transketolase', { 'G3P' 'F6P' }
\path[modStyleDGHyperConnector] (v-1-13) to (v-34-30);
\path[modStyleDGHyperConnector] (v-22-27) to (v-34-30);
\path[modStyleDGHyperConnector] (v-34-30) to (v-5-32);
\path[modStyleDGHyperConnector] (v-34-30) to (v-20-31);
\path[modStyleDGHyperConnector] (v-1-7) to (v-34-38);
\path[modStyleDGHyperConnector] (v-22-36) to (v-34-38);
\path[modStyleDGHyperConnector] (v-34-38) to (v-5-40);
\path[modStyleDGHyperConnector] (v-34-38) to (v-20-39);
\end{tikzpicture}%
}%
%
  	}%
\caption{A realisability certificate for the flow in Fig.~\ref{fig:realisable_flow_PPP_id}. The input compounds are marked with green and the output compounds are marked with blue. The SMILES strings for all molecule identifiers are listed in Appendix, Table~\ref{tab:smiles}.}
\label{fig:realisable_DAG_PPP_id}
\end{figure}

We believe that our focus on the realisability of pathways may help
raise awareness of the choices one often subconsciously makes when
creating pathway illustrations.

\section{Mathematical Properties of Realisability} \label{sec:math_realisability}

In this section, we take the first steps towards a mathematical theory of the realisability of flows. We begin with a result on realisable flows and prove that if the \textit{König representation} of the 
\textit{flow-induced subhypergraph} of the extended hypergraph and flow $f$ does not have any cycles, then $f$ is realisable.

\begin{definition}[Flow-induced Subhypergraph]
The flow-induced subhypergraph of an extended hypergraph $\overline{\mathcal{H}}=(V,\overline{E})$ and a flow $f$ is the directed hypergraph $\overline{\mathcal{H}}[f] = (V', E')$, with
\begin{align}
\begin{aligned}
  E' &=\{e \in E \mid f(e) \neq 0\} \\
  V' &= \{v \in e^+ \lor v\in e^- \mid e\in E, f(e) \neq 0\} \\
\end{aligned}
\end{align}
\end{definition}

\begin{definition}[König Representation \cite{Andersen:20}]
The König representation of a directed hypergraph $\mathcal{H}=(V,E)$ is the directed multigraph $K(\mathcal{H}) = (V', E')$ where $V'=V \cup E$ and 
\begin{align*}
  E' = \mset{(v,e) \mid e=(e^+,e^-)\in E, v\in e^+} \\
  \cup \mset{(e,v) \mid e=(e^+,e^-)\in E, v\in e^-}
\end{align*}
\end{definition}
In short, the König representation of a hypergraph arises simply by
considering both the circles and boxes of its visualization (in the
style of e.g.\ Fig.~\ref{fig:hypergraph}) as nodes and the arrows as
edges.

\begin{lemma}
If $K(\overline{\mathcal{H}}[f])$ has no cycles, then $f$ is realisable.
\end{lemma}

\begin{proof}
Since $K(\overline{\mathcal{H}}[f])$ is a DAG, it has a topological sort. Order the nodes of $\mathcal{H}$ on a line according to this. Create nodes for input (output) ``half-edges'' making them full hyperedges and put these new nodes first (last) in the topological sort. Put the number of tokens specified by the flow on the new input nodes. Create a firing sequence by moving a sweepline across the topological sort and fire a hyperedge when the last node in its source (multi)set is passed. Fire it the number of times specified by the multiplicity of the edge. By the definition of a topological sort, the following holds for any node $v$
\begin{enumerate}[label=(\roman*)]
\item When the sweepline reaches $v$, $v$ has received all its tokens in the flow.
\item Node $v$ only needs to supply tokens after the sweepline has reached $v$.
\item If $v$ is the last node in the sources of a hyperedge, the hyperedge can fire (i.e.\ there are still enough tokens on every node in its source).
\end{enumerate}
Here (i) and (iii) are proven together by induction on the sweepline movements and (ii) is true by construction of the firing sequence.
\end{proof}

There exist flows requiring arbitrary scaling factors:
\begin{theorem}
\label{thm:min-scaled}
For any integer $k > 1$, there exists a flow which is not scaled-realisable for any integer $i<k$ but is scaled-realisable for all integers $i\geq k$.
\end{theorem}
\begin{proof}
One family of such flows is represented by
Fig.~\ref{fig:scaled_large}, which fulfils the statement for $k=4$:
This flow is not scaled-realisable for $i<4$ as all of $B,C,D$ and $E$ need to be present for $r$ to fire in the corresponding Petri net. Therefore, there needs to be at least $4$ tokens input to $A$.
To prove that the flow is scaled-realisable for any integer $i\geq 4$, input $i$ tokens to $A$ and output $i-4$ of them from $A$, such that $4$ tokens still reside on $A$. Fire the sequence $bcder$. There are now $4$ tokens on $A$ again. Repeat the firing of the sequence $bcder$ until it has been fired $k$ times, then output the remaining $4$ tokens on $A$. Clearly, the construction of Fig.~\ref{fig:scaled_large} generalises to any $k>1$. If one would like to avoid the unbounded size of the hyperedge~$r$ in the family, a binary tree structure can be added on both sides of~$r$ (first merging $k$ nodes into one, the expanding this into $k$ nodes connected to $A$).
\end{proof}

\begin{figure}[tbp]
  \centering
  \begin{tikzpicture}[abstractNetwork]
    \node[vertex] (A) {$A$};
    \node[hyperEdge, above right = of A, yshift = -2.6em, label = above : $c$] (a2) {$1$};
    \node[hyperEdge, below right = of A, yshift = 2.6em, label = above : $d$] (a3) {$1$};
    \node[hyperEdge, above right = of A, yshift = -0.3em, label = above : $b$] (a1) {$1$};
    \node[hyperEdge, below right = of A, yshift = 0.3em, label = above : $e$] (a4) {$1$};
    \node[vertex, right = of a1] (B) {$B$};
    \node[vertex, right = of a2] (C) {$C$};
    \node[vertex, right = of a3] (D) {$D$};
    \node[vertex, right = of a4] (E) {$E$};
    \node[hyperEdge, below right = of C, yshift = 3em, label = right : $r$] (r) {$1$};
    
    \node[vertexBase, overlay, left = of A] (A_in) {\phantom{}};
    \node[vertexBase, overlay, above = of A] (A_out) {\phantom{}};
    
    \draw[arrow] (A_in) --  node[above] {$1$} (A);
    \draw[arrow] (A) --  node[left] {$1$} (A_out);
    \draw[arrow] (A) -- (a1);
    \draw[arrow] (A) -- (a2);
    \draw[arrow] (A) -- (a3);
    \draw[arrow] (A) -- (a4);
    \draw[arrow] (a1) -- (B);
    \draw[arrow] (a2) -- (C);
    \draw[arrow] (a3) -- (D);
    \draw[arrow] (a4) -- (E);
    \draw[arrow] (B) -- (r);
    \draw[arrow] (C) -- (r);
    \draw[arrow] (D) -- (r);
    \draw[arrow] (E) -- (r);
    
    \node[vertexBase, overlay, yshift= 2em] at ($(a1)!0.5!(B)$) (mid_top) {\phantom{}};
    \draw[arrow] (r.north) [bend right=39] to (mid_top.center) [bend right=39] to (A);
    
    \node[vertexBase, overlay, above = of mid_top, yshift = -3em] (mid_top_top) {\phantom{}};
    \draw[arrow] (r.north) [bend right=43] to (mid_top_top.center) [bend right=43] to (A);
    
    \node[vertexBase, overlay, yshift= -2em] at ($(a4)!0.5!(E)$) (mid_bottom) {\phantom{}};
    \draw[arrow] (r.south) [bend left=40] to (mid_bottom.center) [bend left=40] to (A);
    
    \node[vertexBase, overlay, below = of mid_bottom, yshift = 3em] (mid_bottom_bottom) {\phantom{}};
    \draw[arrow] (r.south) [bend left=44] to (mid_bottom_bottom.center) [bend left=44] to (A);
  \end{tikzpicture}%
  \caption{A flow which is not scaled-realisable for an integer $k < 4$ but is for any integer $k\geq 4$.}
\label{fig:scaled_large}
\end{figure}

There also exist flows not scaled-realisable for any factor:
\begin{theorem}
 \label{thm:not_scaled}
 The flow in Fig.~\ref{fig:not_scaled} is not scaled-realisable for any integer $k\geq 1$.
\end{theorem}

\begin{figure}[tbp]
  \centering
  \begin{tikzpicture}[abstractNetwork]
    \node[vertex] (A) {$A$};
    \node[hyperEdge, right = of A, label = above : $r$] (r) {$1$};
    \node[vertex, right=of r] (D) {$D$};
    \node[vertexBase, overlay, left = of A] (A_in) {\phantom{}};
    \node[vertexBase, overlay, right = of D] (D_out) {\phantom{}};
    
    \draw[arrow] (A_in) --  node[above] {$1$} (A);
    \draw[arrow] (A) [bend left = 10] to (r);
    \draw[arrow] (A) [bend right = 10] to (r);
    \draw[arrow] (r.south) [bend left = 30] to (A);
    \draw[arrow] (r) -- (D);
    \draw[arrow] (D) --  node[above] {$1$} (D_out);
  \end{tikzpicture}%
  \caption{Simple flow that is not scaled-realisable.}
\label{fig:not_scaled}
\end{figure}

\begin{proof}
Assume that the flow is scaled-realisable for some factor $s$. In the firing sequence that realises the flow, consider the point in time just before the $k$'th firing of $r$. Then at most $s+(k-1)$ tokens on $A$ have been available. For $r$ to now happen the $k$'th time at least $2k$ tokens on $A$ have been available (in order to make the necessary firings of $r$). Hence $s+(k-1) \geq 2k \Leftrightarrow s-1 \geq k$. So $s$ executions of $r$ is not possible.
\end{proof}

We now give an easily checkable condition which if true implies that a flow is not scaled-realisable for any factor. In short, the condition is that at least one vertex of the flow cannot be reached during a graph traversal from its source set.

In more detail: Consider a directed hypergraph $\mathcal{H}=(V,E)$. The
set $R(H,S)$ of vertices reachable from $S$ in $\mathcal{H}$ is
defined by (i) $S\subseteq R(H,S)$ and (ii) if $e^+\subseteq R(H,S)$
for some $e\in E$, then $e^- \subseteq R(H,S)$. It can be computed
using the traversal procedure specified in Alg.~\ref{alg:GT}. Setting
$w=\max_{e\in E} |e^+|$, Alg.~\ref{alg:GT} runs in $O(w|E|^2)$ time,
since checking the condition of the ``while'' loop require $O(w|E|)$
time, and in every iteration, $F$ shrinks by one edge. 

\begin{algorithm}
\SetAlgoLined
\caption{Traversal of a directed hypergraph}
\label{alg:GT}
\DontPrintSemicolon
\KwData{$H=(V,E),S$}
$G\leftarrow S$\\
$F\leftarrow E$

\While{$\exists e\in F\colon e^+\subseteq G$}{
  $G\leftarrow G\cup e^-$\\
  $F\leftarrow F\setminus\{e\}$
}
\Return $G$
\end{algorithm}

\begin{theorem}
  \label{thm:non-scaled}
  Given a flow $f$ on an extended hypergraph
  $\overline{\mathcal{H}}=(V,\overline{E})$, let $\overline{\mathcal{H}}[f]=(V', E')$ be the flow induced subhypergraph and let $S$ be the source set such that
  \begin{align}
    \begin{aligned}
      S  &= \{v\in V \mid f(e_v^-) \neq 0\} \\
    \end{aligned}
  \end{align}
  If there exists a vertex $v \in V'$ that is not returned by the traversal
  procedure on the graph $\overline{\mathcal{H}}[f]$ and source set $S$, then the flow $f$ is not
  scaled-realisable.
\end{theorem}

\begin{proof}
 \label{thm:scaled}
  The flow-induced subhypergraph $\overline{\mathcal{H}}[f]$ has as edges
  the internal edges of $\overline{\mathcal{H}}$ on which there is flow, and
  as vertices the vertices of $\overline{\mathcal{H}}$
  that have either in- or out flow, without regard to the half-edges.
  The source set $S$ contains the vertices of $\overline{\mathcal{H}}$
  with input flow according to $f$. The traversal specified in Alg.~\ref{alg:GT} corresponds
  to having an infinite amount of flow into the vertices in $S$ and no
  restrictions on the number of times an edge can be followed. 
  Therefore, if a vertex is not reachable from
  $S$ by Alg.~\ref{alg:GT}, it is also not reachable in the stricter case,
  where the search is restricted by the flow specification.  Note that the
  omission of the edges on which there is no flow, as well as the vertices
  for which all internal edges entering or leaving them has no flow,
  is crucial in order to let Alg.~\ref{alg:GT} mimic the operations of a
  scaled flow. Otherwise, there might be ways of visiting the vertices, which
  would not be possible if only considering the paths represented
  by the flow specification. Moreover, observe that the omission of
  vertices which only have in- and outflow does not affect the result of
  the algorithm as these would be trivially visited.
\end{proof}

We remark that Thm.~\ref{thm:non-scaled} only provides a sufficient condition for determining non-scaled-realisable flows and not a necessary condition. This follows from the the flow in Fig.~\ref{fig:not_scaled}: during graph traversal, this flow will have all its vertices visited, but by Thm.~\ref{thm:not_scaled} it is not scaled-realisable for any factor.

The property of being scaled-realisable is closed under addition of the scaling factors:
\begin{theorem} \label{thm:scaled_sum}
If a flow $f$ is scaled-realisable for an integer $k$ and an integer $l$, then it is also scaled-realisable for $k+l$.
\end{theorem}
\begin{proof}
Create a realisability certificate for $(k+l)\cdot f$ as the disjoint union of the realisability certificate for $f'= k \cdot f$ and the realisability certificate for $f'' = l \cdot f$.
\end{proof}

The family of flows from the proof of Thm.~\ref{thm:min-scaled} has the following interesting property.
\begin{definition}[Monotone Scaled-Realisable]
A flow $f$ is monotone scaled-realisable iff it is scaled-realisable for all integers $j\geq k$, where $k$ is the smallest factor for which it is scaled-realisable.
\end{definition}
A natural question now arises whether all scaled-realisable flows are also
monotone scaled-realisable.  We did a computer-based search for
counter-examples, but found none.

In more detail, we generated several pseudo-random directed
hypergraphs in which we found a large number of different flows using the software package \textit{MØD} \cite{mod0.5,mod} which has a functionality for executing flow queries for hypergraphs via ILP \cite{flow}. We tested these flows for realisability and among the flows
not directly realisable, we looked at those which were
scaled-realisable with a smallest scale factor $k=2$ or $k=3$.  If the
lowest factor was $k=2$, we tested if the flow was also
scaled-realisable for factor $j=3$.  If the lowest factor was $k=3$,
we tested if the flow was also scaled-realisable for factors $j$ where
$3 < j \leq 5$.  If so, we by Thm.~\ref{thm:scaled_sum} knew that the
flow was monotone scaled-realisable. If not, we would have found a
counter-example. Among the $1688$ scaled-realisable flows studied, we
found them all to be monotone scaled-realisable.

We thus close this section with the following conjecture:
\begin{conjecture}
All scaled-realisable flows are monotone scaled-realisable.
\end{conjecture}

\section{Conclusion}
We introduced here a concept of realisability of a pathway given as an
flow by converting the flow to a Petri net.  The
question of realisability can then be rephrased as a question of
reachability in the Petri net, leading to notions of realisable,
scaled-realisable, and borrow-realisable flows. The method is
essential if one is interested in finding alternative realisable pathways
to those already known by chemists. Reachability in Petri nets and
equivalent formal systems is an active field of research, see
e.g.\ \cite{Alaniz:22} and the references therein. Many of the relevant
reachability problems $M\xrightarrow{*} M'$ are hard for arbitrary
markings. It remains a relevant question for future work to see if
restrictions imposed by chemistry, in particular conservation of mass, suffice to make the problems easier.

An interesting direction for future research is extending the framework to allow for atom tracing in CRNs. While current Petri net methods allows us to track individual tokens/molecules \cite{glabbeek:05}, full atom tracing requires enumerating {\em all} possible firing sequences, i.e., {\em all} witness paths, which existing Petri net tools do not currently provide. On the other hand, atom-atom mapping, i.e., how atoms rearrange during reactions, is already available through an existing graph transformation framework \textit{MØD} \cite{mod0.5,mod}. Such a combination of witness path enumeration and atom-atom mapping is crucial for tracking isotopic labels and understanding reaction mechanisms, and would significantly enhance the model's applicability in systems chemistry, metabolic engineering, and synthetic biology.

\section*{Acknowledgments}
An early version of this paper was published as part of the ISBRA: International Symposium on Bioinformatics Research and Applications (ISBRA 2023) \cite{realisability}.

\section*{Authorship Contribution Statement}
Jakob L. Andersen: Conceptualisation, Methodology, Software, Writing - Original Draft, Writing - Review \& Editing, Supervision. Sissel Banke: Methodology, Writing - Original Draft, Writing - Review \& Editing, Visualisation. Rolf Fagerberg: Conceptualisation, Methodology, Writing - Review \& Editing, Supervision. Christoph Flamm: Conceptualisation. Daniel Merkle: Conceptualisation, Methodology, Writing - Review \& Editing, Supervision. Peter F. Stadler: Conceptualisation, Writing - Review \& Editing.

\section*{Author Disclosure Statement}
The authors have no conflict of interest to declare.

\section*{Funding Information}
This work is supported by the Novo Nordisk Foundation grant NNF19OC0057834 and by the Independent Research Fund Denmark, Natural Sciences, grant DFF-0135-00420B.

\bibliographystyle{plainnat}
\bibliography{paper}

\clearpage
\appendix
\section*{Appendix}
\subsection*{Molecular Structures} \label{sec:structure}
We have omitted the structure of molecules for brevity in some examples in the
paper. This is a large simplification in comparison to the level of detail
handled by our complete framework, which includes a generative approach to
creating chemical reaction networks~\cite{mod0.5}.  Each vertex in the
directed hypergraphs is an undirected labelled graph representing a
molecular structure. Each edge in the directed hyperedge correspond to a
Double Pushout (DPO) transformation of such a graph~\cite{mod0.5}. In
Table~\ref{tab:smiles} we present the correspondence between the ID and
structure of the compounds used throughout the paper.
In Fig.~\ref{fig:DPO} we show an example of a DPO diagram that represents a
reaction. Since the spans in the DPO representation in particular define a
bijection of the vertices (atoms) in tail and head of an hyperedge, it
determines the atom-maps of a reaction.
\begin{figure}[h]
\centering
\renewcommand\modInputPrefix{data/dpo/out}%
\dpoDerivationTwo[scale=0.5]{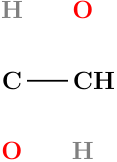}{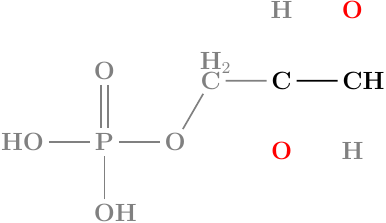}
\caption{Double Pushout diagram for the edge $(\mset{\mathrm{G3P}},
  \mset{\mathrm{DHAP}})$ from Fig.~\ref{fig:borrow_realisable_flow_PPP}.
  Atoms in corresponding locations are mapped onto each in both the rule
  (top) and its application to complete molecules (bottom).}
\label{fig:DPO}
\end{figure}

\begin{sidewaystable}[tbp]
\centering
\caption{SMILES strings of the molecules used throughout the paper.}
\label{tab:smiles}
\begin{tabular}{@{}l@{\hspace{1em}}l@{\hspace{1em}}l@{}}
\toprule
ID  & Name     & SMILES                             \\
\midrule
2Hex     & Hex-2-ulose & \texttt{C(CO)(C(C(C(CO)O)O)O)=O}            \\
Aldohex	& Aldohexose &\texttt{C(C(C(C(C(CO)O)O)O)O)=O}	\\
BuTet	& 1-Butene-1,2,3,4-tetrol &\texttt{C(O)=C(C(CO)O)O}	\\
DHAP   & Dihydroxyacetone phosphate & \texttt{OP(O)(=O)OCC(=O)CO}                 \\
E4	& Threose &\texttt{C(C(C(CO)O)O)=O} 	\\
E4P      &Erythrose 4-phosphate& \texttt{OP(O)(=O)OCC(O)C(O)C=O}             \\
EnolHex	& Enol-hexose  &\texttt{C(O)=C(C(C(C(CO)O)O)O)O}	\\
Erythrulose	&Erythrulose& \texttt{C(CO)(C(CO)O)=O} \\
EtD	& 1,2-Ethenediol &\texttt{C(O)=CO}	\\
F1P      &Fructose 1-Phosphate & \texttt{C(COP(O)(O)=O)(C(C(C(CO)O)O)O)=O}   \\
F6P      &Fructose 6-Phosphate & \texttt{OCC(=O)C(O)C(O)C(O)COP(=O)(O)O}     \\
FBP & Fructose 1,6-bisphosphate & \texttt{OC(COP(O)(O)=O)C(O)C(O)C(COP(O)(O)=O)=O}\\
Formald 	& Formaldehyde &\texttt{C=O}	\\
G3P      & Glyceraldehyde 3-phosphate &\texttt{C(C(C=O)O)OP(=O)(O)O}               \\
Glyald   &Glyceraldehyde &\texttt{C(C(CO)O)=O}                        \\
Glycoald	&Glycolaldehyde& \texttt{OCC=O}	\\
Propenetriol	&Prop-1-ene-1,2,3-triol& \texttt{C(O)=C(CO)O}		\\
R5P      &Ribose 5-phosphate &\texttt{OP(O)(=O)OCC(O)C(O)C(O)C=O}         \\
Ru5P & Ribulose-5-Phosphate / Xylulose 5-Phosphate &\texttt{OCC(=O)C(O)C(O)COP(=O)(O)O}         \\
S7P      &Sedoheptulose 7-phosphate& \texttt{O=P(O)(OCC(O)C(O)C(O)C(O)C(=O)CO)O} \\
\bottomrule
\end{tabular}
\end{sidewaystable}

\clearpage
\subsection*{Scaled-Realisable Formose Pathway with Molecular Structures}
Here we illustrate the scaled-realisable formose pathway from Fig.~\ref{fig:scaled_realisable_flow_formose_id} 
with molecular structures in Fig.~\ref{fig:scaled_realisable_flow_formose_no_id} as well as the realisability certificate 
for said flow from Fig.~\ref{fig:scaled_realisable_DAG_formose_id} with the structure of the molecules in Fig.~\ref{fig:scaled_realisable_DAG_formose_no_id}.

\begin{figure}[h]
\centering
\resizebox{\textwidth}{!}{%
  \renewcommand\modInputPrefix{data/formose_scaled_flow}%
  	\input{data/formose_scaled_flow/out/020_dg_0_01100_f_0_2_filt.tex}%
  	}%
\caption{The flow for the formose reaction from Fig.~\ref{fig:scaled_realisable_flow_formose_id} but with 
molecular structures. It is not realisable but is scaled-realisable by a factor $2$. The input compound Formald is marked with green and Glycoald which is both an input and output compound is marked with turquoise.}
\label{fig:scaled_realisable_flow_formose_no_id}
\end{figure}

\begin{figure}[tbp]
\centering
\resizebox{\textwidth}{!}{%
  \renewcommand\modInputPrefix{data/formose_scaled_cert}%
  	\input{data/formose_scaled_cert/out/021_dg_0_01100.tex}%
  	}%
\caption{The realisability certificate from Fig.~\ref{fig:scaled_realisable_DAG_formose_id} but with the structure of the molecules. The input compounds are marked with green and the output compounds are marked with blue.}
\label{fig:scaled_realisable_DAG_formose_no_id}
\end{figure}

\clearpage

\subsection*{Borrow-Realisable Pentose Phosphate Pathway with Molecular Structures}
In Fig.~\ref{fig:borrow_flow_PPP_no_id} we show the borrow-realisable flow for the PPP from Fig.~\ref{fig:borrow_realisable_flow_PPP} with molecular structures. In Fig.~\ref{fig:borrow_cert_PPP_no_id} we present a certificate for said flow, also with the structure of the molecules. It is the same certificate as the one from Fig.~\ref{fig:borrow_realisable_DAG_PPP}.

\begin{figure}[h]
\centering
\resizebox{0.85\textwidth}{!}{%
  \renewcommand\modInputPrefix{data/ppp_borrow_flow_mol}%
  	\input{data/ppp_borrow_flow_mol/out/022_dg_0_01100_f_0_2_filt.tex}%
  	}%
\caption{The flow for the PPP from Fig.~\ref{fig:borrow_realisable_flow_PPP} but with 
the structure of the molecules. It is not realisable as is but is borrow-realisable. The input compound is marked with green and the output compound is marked with blue.}
\label{fig:borrow_flow_PPP_no_id}
\end{figure}

\begin{figure}[tbp]
\centering
\resizebox{0.95\textwidth}{!}{%
  \renewcommand\modInputPrefix{data/ppp_borrow_cert_mol}%
  	\input{data/ppp_borrow_cert_mol/out/027_dg_0_01100.tex}%
  	}%
\caption{A realisability certificate for the flow in Fig.~\ref{fig:borrow_flow_PPP_no_id} where the compound Glyald is borrowed. It corresponds to the realisability certificate in Fig.~\ref{fig:borrow_realisable_DAG_PPP} but here the molecular structures are visible. The input compounds are marked with green, the output compounds are marked with blue and the borrowed compound is marked with purple.}
\label{fig:borrow_cert_PPP_no_id}
\end{figure}

\clearpage

\subsection*{Realisable Pentose Phosphate Pathway with Molecular structures}
Here we present the flow that depicts a realisable pentose 
phosphate pathway (PPP) from Fig.~\ref{fig:realisable_flow_PPP_id}, but with molecular structures, 
in Fig.~\ref{fig:realisable_flow_PPP} as well 
as a realisability certificate for it, also with molecular structures, which proves its realisability in 
Fig.~\ref{fig:realisable_DAG_PPP}. 

\begin{figure}[h]
  \centering
  \resizebox{0.72\textwidth}{!}{%
  \renewcommand\modInputPrefix{data/ppp_realisable_flow}%
  	\input{data/ppp_realisable_flow/out/022_dg_0_01100_f_0_0_filt.tex}%
  	}%
  \caption{An example of a flow for the pentose phosphate pathway which is realisable. The input compound is marked with green and the output compound is marked with blue.}
  \label{fig:realisable_flow_PPP}
\end{figure}

\begin{figure}[tbp]
  \centering
  \resizebox{0.86\textwidth}{!}{%
  \renewcommand\modInputPrefix{data/ppp_realisable_cert}%
  	\input{data/ppp_realisable_cert/out/023_dg_0_01100.tex}%
  	}%
  \caption{A realisability certificate for the realisable pentose phosphate pathway from Fig.~\ref{fig:realisable_flow_PPP}. The input compounds are marked with green and the output compounds are marked with blue.}
  \label{fig:realisable_DAG_PPP}
\end{figure}
\end{document}